\documentclass[a4paper,11pt]{article} 
\usepackage{amsmath,latexsym, amssymb}
\usepackage{epstopdf,caption,subcaption,graphicx,xcolor,hyperref}
\usepackage{algpseudocode}
\usepackage{algorithm}
\usepackage{tikz}
\usetikzlibrary{shapes,arrows,fit,calc,positioning}
\usetikzlibrary{decorations.pathreplacing}
\normalsize
\newtheorem{theorem}{Theorem}
\newtheorem{lemma}{Lemma}

\newtheorem{observation}{Observation}
\newtheorem{remark}{Remark}

\newtheorem{notation}{Notation}
\newtheorem{definition}{Definition}

\newtheorem{problem}{Problem}
\newenvironment{proof}{{\sc Proof. }}{\hfill$\Box$\vspace{0.1in}}
\newcommand{\mc}[1]{{\cal {#1}}}
\begin{document}

\title{Approximation Algorithms for Perfect Fair-Triangle Packing}

\author{
Mingyang~Gong\thanks{Gianforte School of Computing, Montana State University, Bozeman, MT 59717, USA.
Email: {\tt \{mingyang.gong, brendan.mumey\}@montana.edu}}
\and
Zhi-Zhong~Chen\thanks{Division of Information System Design, Tokyo Denki University. Saitama 350-0394, Japan.
  Email: \texttt{zzchen@mail.dendai.ac.jp}}
\thanks{Correspondence authors.}
\and
Brendan Mumey$^*$$^\ddagger$
}
\date{}
\maketitle

\begin{abstract}
In this paper, we study the {\em perfect fair-triangle packing} problem (abbreviated as PFTP), 
which incorporates the fairness criterion from {\em fair clustering} into the {\em maximum-weight triangle packing} problem.
Specifically, the input is an edge-weighted complete graph $G = (V, E)$ with $|V| = 3n$, where each vertex is colored red or blue.
A {\em fair triangle} is a triangle containing vertices of both colors.
PFTP asks for a partition of $V$ into $n$ fair triangles such that the total edge weight is maximized.
To the best of our knowledge, this is the first paper to study PFTP.

PFTP is NP-hard.
Our main contributions are a deterministic $\frac 13$-approximation algorithm running in $O(n^3)$ time 
and an improved randomized $(\frac {16}{47}-\epsilon)$-approximation algorithm running in $O(n^4)$ time,
where $\epsilon > 0$ is a fixed small constant.
The deterministic algorithm is matching-based whereas the randomized algorithm employs several additional techniques, 
including maximum-weight $[1, 2]$-factor, a random cycle-breaking procedure, and maximum-weight matchings.
\\\\
\textbf{Keywords: } Triangle packing; fairness; approximation algorithms; randomized algorithms 
\end{abstract}

\newpage

\section{Introduction}\label{sec:intro}

\subsection{Model and motivation}

In the {\em maximum-weight triangle packing} problem (abbreviated as MWTP), we are given an edge-weighted complete graph $G = (V, E)$ with $|V| = 3n$.
MWTP asks for a collection of $n$ vertex-disjoint triangles with maximum total edge weight.
As a fundamental combinatorial optimization problem, MWTP has been extensively studied in the literature,
both for arbitrary edge weights~\cite{HR06, CTW09, CTW10}
and for the metric setting, where the edge weights satisfy the triangle inequality~\cite{CCL21, ZX24}.

{\em Fair clustering}~\cite{CKL17, BCF19, BGK19} is another fundamental problem arising in both operations research and machine learning.
Let $G = (V, E)$ be an edge-weighted graph (unlike in MWTP, $G$ is not necessarily complete), where each vertex is colored red or blue.
In the seminal work~\cite{CKL17}, given an integer $t \ge 1$, the authors define a {\em fair cluster} as a cluster 
whose number of red vertices is between $\frac{1}{t}$ and $t$ times the number of blue vertices.
The term ``fair'' refers to the parameter $t$, which ensures that the numbers of red and blue vertices remain approximately balanced.
The fair clustering problem seeks $k$ vertices as centers and assigns every other vertex to a center 
so that each resulting cluster is fair while optimizing an objective such as the $k$-center or $k$-median objective.

Suppose that $G = (V, E)$ is an edge-weighted complete graph with $|V| = 3n$, where each vertex is colored either red or blue. 
Motivated by fair clustering~\cite{CKL17}, we define {\em a fair triangle} as a triangle containing vertices of both colors.
Consequently, a fair triangle contains either one or two red vertices (equivalently, $t \ge 2$ in the setting of~\cite{CKL17}).
The two-color setting of a fair triangle naturally models various real-world applications, such as gender diversity.
Another example arises in the formation of a Ph.D. defense committee.
In such a committee, supervisors who have previously collaborated with the candidate naturally serve as internal members.
However, the committee must also include at least one external reviewer who has had no prior collaboration with the candidate, 
thereby ensuring an objective evaluation of the candidate's work and helping to maintain the fairness of the examination.
These two types of reviewers can be represented by the red and blue vertices in a fair triangle. 

Recall $G=(V,E)$ introduced in the preceding paragraph. 
We define {\em a perfect fair-triangle packing} in $G$ to be a partition of $V$ into $n$ fair triangles.
Inspired by MWTP and fair clustering, we study a variant of MWTP, called the {\em perfect fair-triangle packing} problem (abbreviated as PFTP), 
which seeks a perfect fair-triangle packing in $G$ with maximum total edge weight. 
Let $r$ be the number of red vertices in $G$ and assume without loss of generality that $r \le \frac 32 n$.
As discussed in Section~2, $G$ admits a feasible solution to PFTP if and only if $n \le r \le \frac 32n$.
Therefore, we only consider input graphs satisfying $n \le r \le \frac 32 n$.

In this paper, we study PFTP from the perspective of approximation algorithms and 
assume that the reader is familiar with standard concepts in approximation algorithms, such as approximation ratio~\cite{WS11}.

\subsection{Related work}
\label{sec1.2}

PFTP incorporates the fairness criterion from fair clustering into MWTP.
Accordingly, we review related problems and results that are relevant to PFTP.

\noindent
{\bf Unweighted maximum triangle packing (MTP): } 
The {\em unweighted maximum triangle packing} problem (MTP) seeks a largest collection of vertex-disjoint triangles 
in a given (not necessarily complete) graph $G$. 
Its decision version, the {\em partition into triangles} problem~\cite{GJ79}, asks whether vertex-disjoint triangles can cover all vertices of $G$.
This decision problem is NP-complete~\cite{GJ79} via a reduction from {\em $3$-dimensional matching}, rendering MTP NP-hard.
Moreover, MTP is APX-hard even restricted to graphs of maximum degree $4$~\cite{Kan91, RNB13}.
Chleb\'{i}k and Chleb\'{i}kov\'{a} further showed that, unless P=NP, MTP cannot be approximated within $\frac {94}{95} \approx 0.9895$~\cite{CC06} in polynomial time -- 
currently the strongest inapproximability bound.

In the {\em unweighted $3$-set packing} problem (U$3$SP), we are given a collection of subsets of a universe, 
each containing at most $3$ elements.
The objective is to select a maximum number of pairwise disjoint subsets in the collection.
Since MTP is a special case of U$3$SP~\cite{HR06, CTW09}, any approximation bound for U$3$SP applies to MTP as well. 
For any fixed $\epsilon>0$, U$3$SP can be approximated within $\frac 34-\epsilon$ in polynomial time~\cite{Cyg13, FY14}, 
thereby yielding a $(\frac 34-\epsilon)$-approximation for MTP.

\noindent
{\bf Maximum-weight triangle packing (MWTP): } 
Unlike MTP, the {\em maximum weight triangle packing} (MWTP) problem considers an edge-weighted complete graph $G=(V,E)$ with $3n$ vertices, 
seeking a partition of $V$ into $n$ triangles with maximum total edge weight.
For any fixed constant $\epsilon > 0$, Hassin and Rubinstein~\cite{HR06, HR06a} gave a randomized 
$(0.518-\epsilon)$-approximation algorithm.
Chen {\em et al.}~\cite{CTW09, CTW10} later improved this to $(0.523-\epsilon)$, a ratio also matched deterministically by van Zuylen~\cite{Zuy13}. 
Under the metric setting (where edge weights satisfy the triangle inequality), 
Chen {\em et al.}~\cite{CCL21} designed a $(0.66768-\epsilon)$-approximation, which Zhao and Xiao~\cite{ZX24} recently improved to $0.66835-\epsilon$.

The {\em weighted $3$-set packing} problem (W$3$SP) generalizes U$3$SP by assigning a non-negative weight to each candidate subset, 
aiming to maximize the total weight of a pairwise disjoint subcollection.
Since MWTP is a special case of W$3$SP~\cite{HR06, CTW09}, any bound for W$3$SP carries over. 
The best known approximation ratio for W$3$SP is $\frac 1{1.786} \approx 0.5599$ 
by Thiery and Ward~\cite{TW23}, which provides the current best approximation guarantee of $0.5599$ for MWTP.

\noindent
{\bf Fair clustering: }
Given an integer $t \ge 1$ and an edge-weighted graph $G$ whose vertices are colored red or blue, Chierichetti {\em et al.}~\cite{CKL17} introduced the notion of {\em fair clusters}, 
defining a cluster to be {\em fair} if its ratio of red to blue vertices lies in $[\frac 1t, t]$.
They obtained polynomial-time $4$-approximation and $O(t)$-approximation algorithms for the $k$-center and $k$-median objectives, respectively~\cite{CKL17}.
Bandyapadhyay {\em et al.}~\cite{BBC25} later gave a polynomial-time $O(1)$-approximation algorithm for the {\em sum-of-radii} objective. 

Bercea {\em et al.}~\cite{BGK19} and Bera {\em et al.}~\cite{BCF19} independently generalized fair clustering to $h\ge 2$ colors.
Here, given an edge-weighted graph $G=(V, E)$ with an $h$-coloring on $V$, a cluster is {\em fair} 
if the proportion of color-$j$ vertices lies within the interval $[\ell_j, u_j]$ for each $j \in [1, h]$, where $0 \le \ell_j \le u_j \le 1$.
In machine learning, $\ell_j$ and $u_j$ are typically chosen close to the overall proportion of color $j$ in $V$~\cite{DEM25}.
Approximation algorithms under various clustering objectives have since been 
developed in~\cite{BCF19, BGK19, BBC25, BCF25, BCF25a}.

\noindent
{\bf Relationship between PFTP and W$3$SP: } 
Recall that MWTP is a special case of W$3$SP, meaning any algorithm for W$3$SP can be applied to MWTP.
A natural question is whether this relationship extends to PFTP.
We next show that PFTP reduces to W$3$SP when $r = n$, but not for general values of $r$, 
where the input graph has exactly $3n$ vertices, $r$ of which are red.

When $r = n$, every triangle in a perfect fair-triangle packing contains exactly one red vertex and two blue vertices.
Given a PFTP instance $G$, we construct a W$3$SP instance $I$ by creating a set for each 
fair triangle with exactly one red vertex, assigning it a weight equal to the triangle's total edge weight. 
Any perfect fair-triangle packing in $G$ maps to $n$ disjoint sets in $I$ of equal total weight.
Conversely, any solution $\mc{L}$ to $I$ corresponds to $|\mc{L}| = \ell$ disjoint fair triangles in $G$. 
The remaining $n-\ell$ red vertices and $2n-2\ell$ blue vertices can be arbitrarily grouped into $n-\ell$ fair triangles, 
producing a perfect fair-triangle packing with weight at least that of $\mc{L}$. 
Thus, PFTP with $r=n$ reduces to W$3$SP.

For arbitrary $r$, however, W$3$SP algorithms cannot be directly applied to PFTP. 
Consider a graph $G$ with four red vertices $v_1, \ldots, v_4$ and five blue vertices $v_5, \ldots, v_9$.
If we construct a W$3$SP instance $I$ containing all fair triangles in $G$, a W$3$SP algorithm might select two sets containing $\{v_1, v_2\}$ and $\{v_3, v_4\}$ respectively 
if edges between red vertices have sufficiently large weights. 
Converting this solution back to a packing in $G$ leaves three blue vertices, which cannot form a fair triangle. 
Hence, PFTP cannot be solved by a direct application of W$3$SP algorithms.

\subsection{Our results and approaches}

Recall that each PFTP instance is defined on an edge-weighted complete graph $G = (V, E)$ with $|V| = 3n$, 
where each vertex is colored red or blue, and the number $r$ of red vertices in $G$ satisfies $n \le r \le \frac 32n$.
An edge is called {\em red} (respectively, {\em blue}) if both its endpoints are red (respectively, blue).
In contrast, a {\em bichromatic} edge connects a red vertex to a blue vertex.

PFTP is NP-hard  even on instances with $r = n$ (see Theorem~\ref{thm01}), 
which we establish via a simple reduction from the $3$-dimensional matching problem.
Our main contributions are two approximation algorithms for PFTP:
\begin{itemize}
\parskip=0pt
\item[1.] A deterministic $\frac 13$-approximation algorithm running in $O(n^3)$ time. 
\item[2.] A randomized $(\frac {16}{47}-\epsilon)$-approximation algorithm running in $O(n^4)$ time, 
for any fixed small constant $\epsilon > 0$.
\end{itemize}
Let $R$, $B$, and $X$ denote the red, blue, and bichromatic edge sets of $G$.
Moreover, let $G_{\rm r}$, $G_{\rm b}$, and $G_{\rm x}$ denote the spanning subgraphs $(V, R)$, $(V, B)$, and $(V, X)$, respectively.
The $\frac 13$-approximation algorithm returns the heavier of two packings, $\mc{T}_0$ and $\mc{T}_1$. 
To construct $\mc{T}_0$, it computes maximum-weight matchings $M_{\rm r}$ of size $r-n$ in $G_{\rm r}$ and $M_{\rm b}$ of size $2n-r$ in $G_{\rm b}$, 
then forms triangles by pairing unmatched red (respectively, blue) vertices with edges in $M_{\rm b}$ (respectively, $M_{\rm r}$). 
Analogously, $\mc{T}_1$ is formed by finding a maximum-weight matching $M_{\rm x}$ of size $n$ in $G_{\rm x}$ and pairing each remaining vertex with an edge in $M_{\rm x}$. 

The randomized $(\frac {16}{47}-\epsilon)$-approximation algorithm starts with the perfect fair-triangle packing 
$\mc{T}_0$ produced by the $\frac 13$-approximation algorithm, 
computes four additional perfect fair-triangle packings $\mc{T}_1, \ldots, \mc{T}_4$, and outputs the heaviest one. 
Unlike our deterministic algorithm and the existing algorithms for MWTP (see Section~\ref{sec1.2}), 
$\mc{T}_1, \ldots, \mc{T}_4$ are based on a maximum-weight $[1, 2]$-factor $F$ in $G_{\rm x}$~\cite{Gab83}. 
Let $\mc{F}=(V, F)$. 
Each connected component of $\mc{F}$ with at least $\frac 2\epsilon$ edges is broken 
into paths of length at most $\frac 2\epsilon-1$, losing at most an $\epsilon$ fraction of its total weight. 
Afterwards, the bichromatic-edge set $X$ is partitioned into five subsets $X_1, \ldots, X_5$, according to the colors of 
their endpoints and the connected components of $\mc{F}$ to which they belong. 
Roughly speaking, for each $i \in \{1, \ldots, 4\}$, the packing $\mc{T}_i$ is  
constructed so that its total edge weight is guaranteed to be relatively large whenever a fixed optimal solution contains 
a heavy subset of $X_i$. 
Furthermore, both $\mc{T}_3$ and $\mc{T}_4$ extend this guarantee to cases where the optimal solution has a heavy subset of $X_5$.
The primary technical challenge is to construct each of $\mc{T}_1, \ldots, \mc{T}_4$ as a perfect fair-triangle packing, rather than merely a fair-triangle packing.  
The packing $\mc{T}_1$ is obtained via Lemma~\ref{lemma03} from what we call a feasible fair $(1, 2)$-path packing.
To construct $\mc{T}_2$, $\mc{T}_3$, and $\mc{T}_4$, we return to the initial $[1, 2]$-factor $\mc{F}$ in $G_{\rm x}$, 
randomly decompose its cycles into odd alternating paths (Definition~\ref{def02}),
and use its paths to construct another collection of alternating or special paths (Definitions~\ref{def02} and~\ref{def03}).
These paths are subsequently connected using certain edges from a maximum-weight matching in $(V, X_i)$ with $i \in \{ 2, 3, 4 \}$.
Finally, applying Lemmas~\ref{lemma05} and~\ref{lemma06} to the resulting paths always yields the desired perfect fair-triangle packings.

The remainder of this paper is organized as follows.
Section~\ref{sec2} gives basic definitions.
Section~\ref{sec2+} presents the formal definition of PFTP and establishes its NP-hardness.
Sections~\ref{sec:simpleAlgo} and~\ref{sec:random} present the deterministic $O(n^3)$-time $\frac 13$-approximation algorithm 
and the randomized $O(n^4)$-time $(\frac {16}{47}-\epsilon)$-approximation algorithm, respectively.
Finally, Section~\ref{sec:conclude} concludes the paper with directions for future research.

\section{Basic definitions}
\label{sec2}

Throughout the remainder of this paper, a graph means a simple undirected graph (i.e., it has neither parallel edges nor self-loops).
Let $G$ be a graph.
We denote the vertex set of $G$ by $V(G)$, and denote the edge set of $G$ by $E(G)$. 
The {\em degree} of a vertex $v$ in $G$, denoted by $d_G(v)$, is the number of vertices adjacent to $v$ in $G$. 
For a subset $U$ of $V(G)$, the subgraph {\em induced by $U$} is $(U, E_U)$, 
where $E_U$ consists of all edges $\{u,v\}$ of $G$ with $u \in U$ and $v\in U$. 
For a subset $F$ of $E(G)$, let $G[F]$ denote the spanning subgraph $(V(G), F)$ of $G$. 

A {\em cycle} in $G$ is a connected subgraph of $G$ in which each vertex is of degree~$2$. 
A {\em path} in $G$ is a connected subgraph of $G$ in which {\em exactly} two vertices (called {\em endpoints}) are of degree~$1$ and all other vertices are of degree~$2$. 
Note that, unlike the standard definition, a single vertex is not considered a path. 
The {\em length} of a cycle or path $C$ is the number of edges in $C$. 
A {\em $k$-cycle} (respectively, {\em $k$-path}) is a cycle (respectively, path) of length~$k$. 
In particular, a $3$-cycle is also called a {\em triangle}. 
For simplicity, we denote a triangle with vertices $u, v, t$ by $uvt$. 

A {\em path} (respectively, {\em cycle}) {\em component} of $G$ is a connected component that is a path 
(respectively, cycle). If a path component is an edge, we call it an {\em edge component}. 

Suppose that each edge of $G$ has a nonnegative weight. 
For an edge $e \in E(G)$, $w(e)$ denotes its weight. 
The {\em weight} of a subset $E' \subseteq E(G)$, denoted by $w(E')$, is $\sum_{e \in E'} w(e)$. 
The {\em weight} of $G$, denoted by $w(G)$, is $w(E(G))$. 
A {\em matching} in $G$ is a set $M$ of pairwise vertex-disjoint edges in $G$; 
hence $w(M)$ denotes its weight.

Suppose further that $G$ is {\em bicolored}; that is, each vertex of G is colored either red or blue. 
An edge of $G$ is {\em red} (respectively, {\em blue}) if both of its endpoints are colored red (respectively, blue); 
otherwise, it is {\em bichromatic}. 
In Figure~\ref{fig01}, the solid edges are bichromatic, while the dashed edge in the left (respectively, right) triangle 
is red (respectively, blue). 
A {\em red} (respectively, {\em blue} or {\em bichromatic}) {\em matching} in $G$ is a matching 
whose edges are all red (respectively, blue or bichromatic).

A triangle in $G$ is {\em fair} if it contains at least one bichromatic edge. 
Every fair triangle contains exactly one red edge or exactly one blue edge (see Figure~\ref{fig01}). 
It is called {\em red-dominant} in the former case and {\em blue-dominant} in the latter.
A {\em fair-triangle packing} of $G$ is a set $\mc{T}$ of pairwise vertex-disjoint fair triangles in $G$. 
We regard $\mc{T}$ as a subgraph of $G$, and hence $w(\mc{T})$ denotes its weight. 
$\mc{T}$ is {\em perfect} if every vertex of $G$ belongs to a triangle in $\mc{T}$. 
If $\mc{T}$ is perfect and $w(\mc{T})$ is maximum among all perfect fair-triangle packings of $G$, 
then $\mc{T}$ is called an {\em optimal perfect fair-triangle packing} of $G$. 

For a random variable $Y$, $\mathbb{E}[Y]$ denotes its expected value. Similarly, 
for an event $A$, $\Pr[A]$ denotes its probability. 
For two events $A$ and $B$, $\Pr[A~|~B]$ denotes the conditional probability of $A$ given $B$. 

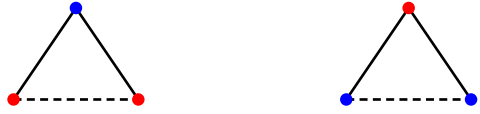
\begin{figure}[thb]
\begin{center}
\begin{tikzpicture}[scale=0.55,transform shape]

\foreach \x in { 0, 8} {
         \draw[thick, line width = 1pt] (\x, 0) -- (\x+1.5, -2.2);
         \draw[thick, line width = 1pt] (\x, 0) -- (\x-1.5, -2.2);
}

\draw[densely dashed, line width = 1pt] (-1.5, -2.2) -- (1.5, -2.2);
\draw[densely dashed, line width = 1pt] (6.5, -2.2) -- (9.5, -2.2);

\fill[blue] (0, 0) circle(.15);
\fill[red] (1.5, -2.2) circle(.15);
\fill[red] (-1.5, -2.2) circle(.15);
\fill[red] (8, 0) circle(.15);
\fill[blue] (9.5, -2.2) circle(.15);
\fill[blue] (6.5, -2.2) circle(.15);

\end{tikzpicture}
\end{center}
\caption{A red-dominant triangle (left) and a blue-dominant triangle (right).
All solid edges are bichromatic.
The dashed edge in the left (respectively, right) triangle is red (respectively, blue).
\label{fig01}}
\end{figure}

\section{The perfect fair-triangle packing problem (PFTP)}
\label{sec2+}
%
Given a bicolored, edge-weighted complete graph $G$, we are interested in finding an optimal perfect fair-triangle packing in $G$. 
Under what conditions does such a packing exist? 

An obvious necessary condition is that $|V(G)|$ is divisible by~$3$. 
Henceforth, we assume that $|V(G)| = 3n$ for some positive integer $n$.
Let $r$ be the number of red vertices in $G$. 
Then, the number of blue vertices is $3n-r$.
Without loss of generality, we may assume that $r \le 3n-r$, or equivalently, $r \le \frac 32 n$.
A simple calculation shows that every perfect fair-triangle packing of $G$ contains exactly 
$r-n$ red-dominant triangles and exactly $2n-r$ blue-dominant triangles.
Hence, $r\ge n$ is a necessary condition for the existence of a perfect fair-triangle packing in $G$. 
In fact, since $G$ is complete, this condition is also sufficient.

We now formally define the problem studied in this paper.

\begin{problem}
\label{prob01}
We are given a bicolored, edge-weighted complete graph $G$ with $r$ red vertices and $3n-r$ blue vertices, 
where $n \le r \le \frac 32 n$. 
PFTP is the problem of computing an optimal perfect fair-triangle packing of $G$.
\end{problem}

\begin{observation}
\label{obs:1}
Let $\mc{T}$ be a perfect fair-triangle packing of $G$. 
Clearly, the red edges in $\mc{T}$ form a red matching of size $r-n$ and 
the blue edges in $\mc{T}$ form a blue matching of size $2n-r$. 
Moreover, $\mc{T}$ has exactly $2n$ bichromatic edges, which form $n$ pairwise vertex-disjoint $2$-paths.
\end{observation}

We next recall the {\em $3$-dimensional matching problem} ($3$-DM)~\cite{GJ79}. 
Given three pairwise disjoint sets $S$, $W$, and $Z$ of equal size, together with a collection $T \subseteq S \times W \times Z$ of triples, 
$3$-DM asks whether there exists a subset $T'\subseteq T$ such that $|T'|=|S|$ and no two triples in $T'$ share an element. 
It is known that $3$-DM is NP-complete~\cite{GJ79}.

\begin{theorem}
\label{thm01}
PFTP is NP-hard even when the input graph $G$ has exactly $\frac 13 |V(G)|$ red vertices 
(i.e., $r = n$ in Problem~\ref{prob01}).
\end{theorem}
\begin{proof}
The proof is by reduction from 3-DM and closely follows the proof of the NP-hardness of  the 
{\em partition into triangles} problem in~\cite{GJ79}. 
Given an instance $I=(S,W,Z,T)$ of $3$-DM, we construct a graph $G$ as follows:
First, for each element $a \in S \cup W \cup Z$, we add a corresponding vertex $\hat{a}$ to $G$.
If $a \in S$, then we color $\hat{a}$ red; otherwise, we color $\hat{a}$ blue.
We then process the triples in $T$ one by one.
For each triple $\tau = (s, w, z) \in T$, we add three red vertices $\hat{\tau}_7, \hat{\tau}_8, \hat{\tau}_9$ and 
six blue vertices $\hat{\tau}_1, \ldots, \hat{\tau}_6$ to $G$ as shown in Figure~\ref{fig02}; 
we also add the $18$ edges shown in Figure~\ref{fig02} to $G$, each with weight $1$.
After processing the triples in $T$, we complete $G$ by adding the remaining edges,
each with weight $0$. 
Clearly, $G$ contains exactly $3|S|+9|T|$ vertices, of which $|S|+3|T|$ are red.
It suffices to show that $I$ is a yes-instance if and only if $G$ admits a perfect fair-triangle packing $\mc{T}$ with $w(\mc{T}) = 3|S|+9|T|$, as we do below.

Suppose that $I$ is a yes instance. 
Then, $T$ has a subset $T'$ in which each element of $S \cup W\cup Z$ appears exactly once. 
So, $|T'| = |S|$. 
We construct $\mc{T}$ from $I$ and $T'$ as follows. 
For each triple $\tau = (s, w, z) \in T'$, $\mc{T}$ contains the four fair triangles 
$\hat{s}\hat{\tau}_1\hat{\tau}_2$, $\hat{w}\hat{\tau}_3\hat{\tau}_8$,  $\hat{z}\hat{\tau}_4\hat{\tau}_9$, and $\hat{\tau}_5\hat{\tau}_6\hat{\tau}_7$.
Similarly, for each triple $\tau = (s, w, z) \in T \setminus T'$, $\mc{T}$ contains the three fair triangles $\hat{\tau}_1\hat{\tau}_2\hat{\tau}_7$,  $\hat{\tau}_3\hat{\tau}_5\hat{\tau}_8$, and $\hat{\tau}_4\hat{\tau}_6\hat{\tau}_9$, as shown in Figure~\ref{fig02}.
In total, $\mc{T}$ contains $4|T'|+3|T \setminus T'| = |S|+3|T|$ pairwise vertex-disjoint fair triangles.
Since $|V(G)| = 3|S|+9|T|$, the triangles in $\mc{T}$ form a perfect fair-triangle packing of $G$. 
Moreover, every edge used by $\mc{T}$ has weight~$1$, and hence $w(\mc{T}) = 3|S|+9|T|$.
This completes the proof of the ``only-if" direction. 

Conversely, suppose that $G$ has a fair triangle packing $\mc{T}$ with $w(\mc{T}) = 3|S|+9|T|$.
Since $w(\mc{T}) = 3|S|+9|T|$ and $G$ has exactly $3|S|+9|T|$ vertices, every edge in $\mc{T}$ has weight $1$.
For each element $s \in S$, consider the triangle in $\mc{T}$ containing its corresponding vertex $\hat{s}$.
Since $w(e) = 1$ for every $e \in E(\mc{T})$, there exists a triple $\tau = (s, w, z) \in T$ such that the vertices $\hat{s}$, $\hat{\tau}_1$ and $\hat{\tau}_2$ form a triangle in $\mc{T}$.
All such triples $\tau$ form a set $T' \subseteq T$.
Clearly, $|T'| = |S|$.
To finish the proof, it suffices to show that for any distinct $s, s' \in S$, 
the corresponding $\tau = (s, w, z)$ and $\tau' = (s', w', z')$ in $T'$ satisfy both $w \ne w'$ and $z \ne z'$, or 
equivalently, both $\hat{w} \ne \hat{w}'$ and $\hat{z} \ne \hat{z}'$.
Recall that $\hat{s}\hat{\tau}_1\hat{\tau}_2$ is a triangle in $\mc{T}$.
By this fact, the edges $\{ \hat{\tau}_1, \hat{\tau}_7 \}$ and $\{ \hat{\tau}_2, \hat{\tau}_7 \}$ are not in $\mc{T}$.
Consequently, the triangle in $\mc{T}$ containing $\hat{\tau}_7$ must be $\hat{\tau}_5\hat{\tau}_6\hat{\tau}_7$.
By the same argument, we can show $\hat{w}\hat{\tau}_3\hat{\tau}_8$ and $\hat{z}\hat{\tau}_4\hat{\tau}_9$ are triangles in $\mc{T}$, and so are $\hat{w}'\hat{\tau}'_3\hat{\tau}'_8$ and $\hat{z}'\hat{\tau}'_4\hat{\tau}'_9$.
In summary, the vertices $\hat{w}$, $\hat{w}'$, $\hat{z}$, and $\hat{z}'$ belong to distinct triangles in $\mc{T}$.
Therefore, $\hat{w} \ne \hat{w}'$ and $\hat{z} \ne \hat{z}'$.
This completes the proof.
\end{proof}

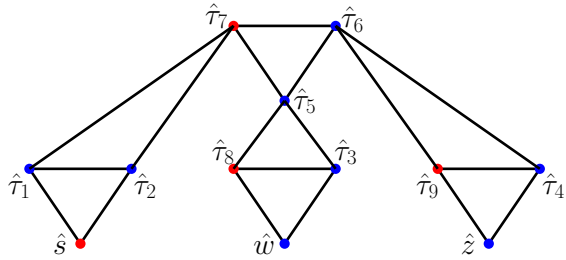
\begin{figure}[thb]
\begin{center}
\begin{tikzpicture}[scale=0.45,transform shape]

\draw[thick, line width = 1pt] (0, -2.2) -- (1.5, 0);
\draw[thick, line width = 1pt] (0, -2.2) -- (-1.5, 0);
\draw[thick, line width = 1pt] (1.5, 0) -- (-1.5, 0);
\fill[red] (0, -2.2) circle(.15);
\fill[blue] (1.5, 0) circle(.15);
\fill[blue] (-1.5, 0) circle(.15);

\draw[thick, line width = 1pt] (6, 2) -- (7.5, 4.2);
\draw[thick, line width = 1pt] (6, 2) -- (4.5, 4.2);
\draw[thick, line width = 1pt] (7.5, 4.2) -- (4.5, 4.2);
\fill[blue] (6, 2) circle(.15);
\fill[blue] (7.5, 4.2) circle(.15);
\fill[red] (4.5, 4.2) circle(.15);
         
\foreach \x in { 6, 12 } {
         \draw[thick, line width = 1pt] (\x, -2.2) -- (\x+1.5, 0);
         \draw[thick, line width = 1pt] (\x, -2.2) -- (\x-1.5, 0);
         \draw[thick, line width = 1pt] (\x+1.5, 0) -- (\x-1.5, 0);
         \fill[blue] (\x, -2.2) circle(.15);
         \fill[red] (\x-1.5, 0) circle(.15);
         \fill[blue] (\x+1.5, 0) circle(.15);
}

\draw[thick, line width = 1pt] (4.5, 4.2) -- (-1.5, 0);
\draw[thick, line width = 1pt] (4.5, 4.2) -- (1.5, 0);
\draw[thick, line width = 1pt] (6, 2) -- (4.5, 0);
\draw[thick, line width = 1pt] (6, 2) -- (7.5, 0);
\draw[thick, line width = 1pt] (7.5, 4.2) -- (10.5, 0);
\draw[thick, line width = 1pt] (7.5, 4.2) -- (13.5, 0);

\node[font=\fontsize{24}{6}\selectfont] at (-0.6, -2.2) {$\hat{s}$};
\node[font=\fontsize{24}{6}\selectfont] at (-1.8, -0.5) {$\hat{\tau}_1$};
\node[font=\fontsize{24}{6}\selectfont] at (1.9, -0.5) {$\hat{\tau}_2$};

\node[font=\fontsize{24}{6}\selectfont] at (5.4, -2.2) {$\hat{w}$};
\node[font=\fontsize{24}{6}\selectfont] at (4.2, 0.5) {$\hat{\tau}_8$};
\node[font=\fontsize{24}{6}\selectfont] at (7.8, 0.5) {$\hat{\tau}_3$};

\node[font=\fontsize{24}{6}\selectfont] at (11.4, -2.2) {$\hat{z}$};
\node[font=\fontsize{24}{6}\selectfont] at (10.2, -0.5) {$\hat{\tau}_9$};
\node[font=\fontsize{24}{6}\selectfont] at (13.9, -0.5) {$\hat{\tau}_4$};

\node[font=\fontsize{24}{6}\selectfont] at (6.6, 2) {$\hat{\tau}_5$};
\node[font=\fontsize{24}{6}\selectfont] at (4, 4.5) {$\hat{\tau}_7$};
\node[font=\fontsize{24}{6}\selectfont] at (8, 4.4) {$\hat{\tau}_6$};
\end{tikzpicture}
\end{center}
\caption{The gadget constructed for a given triple $\tau = (s, w, z) \in T$. 
Only the edges of weight $1$ are shown; all remaining edges have weight $0$.
\label{fig02}}
\end{figure}

\section{A simple $\frac 13$-approximation algorithm for PFTP}\label{sec:simpleAlgo}

Before we present the $\frac 13$-approximation algorithm for PFTP, we recall the {\em maximum-weight size-$k$ matching} problem. 
Given an edge-weighted graph $G_1 = (V_1, E_1)$ and a non-negative integer $k$, 
the problem seeks a matching in $G_1$ of size $k$ that maximizes the total edge weight.
We refer to such a matching as a {\em maximum-weight size-$k$ matching} in $G_1$.
The next lemma proves that a maximum-weight size-$k$ matching can be computed in $O(|V_1|^3)$
and this result will be employed as a subroutine in our algorithm for PFTP.

\begin{lemma}
\label{lemma01}
{\em \cite{Gab76, Gab83}}
Given an edge-weighted graph $G_1 = (V_1, E_1)$ and a non-negative integer $k$, 
a maximum-weight size-$k$ matching in $G_1$ can be found in $O(|V_1|^3)$.
\end{lemma}
\begin{proof}
Given the graph $G_1$, we construct an auxiliary graph $G_2$ by adding an additional vertex $u \notin V_1$
and connecting $u$ to every vertex $v \in V_1$ with an edge of weight $0$.
For each vertex $v \in V_1$, we define $f(v) = g(v) = 1$ and $f(u) = g(u) = |V_1|-2k$.
An {\em [f, g]-factor} of $G_2$ is a subset $F_2 \subseteq E(G_2)$ such that in the spanning subgraph $(V(G_2), F_2)$, 
the degree of each vertex $v$ is between $f(v)$ and $g(v)$, inclusive.
The {\em weight} of an $[f, g]$-factor $F_2$ is the total weight of its edges.

We compute a maximum-weight $[f, g]$-factor $F_2$ in $G_2$ in $O(|V_1|^3)$ time~\cite{Gab83}.
Obviously, the total degree of vertices in $(V(G_2), F_2)$ is $2|F_2|$, and hence $|F_2| = |V_1| - k$.
By the definitions of $f$ and $g$, exactly $|V_1|-2k$ edges in $F_2$ are incident to $u$
and exactly one edge in $F_2$ is incident to each vertex $v \in V_1$.
After removing the $|V_1|-2k$ edges incident to $u$ from $F_2$, we obtain a set $F_1 \subseteq E_1$ of edges with $|F_1| = k$.
Moreover, the degree of each $v \in V_1$ in $(V_1, F_1)$ is either $0$ or $1$ depending on whether the edge $\{ u, v \}$ is in $F_2$.
Consequently, $F_1$ is a matching of size $k$ in $G_1$.

It suffices to show that $F_1$ is a maximum-weight size-$k$ matching in $G_1$.
Assume, to the contrary, that there exists a matching $F'_1$ of size $k$ in $G_1$ whose weight exceeds that of $F_1$. 
Then, exactly $|V_1|-2k$ vertices of $G_1$ have degree $0$ in the graph $(V_1, F'_1)$. 
We construct a new set $F'_2$ from $F'_1$ by adding the edge $\{ u, v \}$ for all $v \notin V(F'_1)$.
Clearly, $F'_2$ is an $[f, g]$-factor of $G_2$.
Moreover, $w(F_1) = w(F'_1)$ and $w(F_2) = w(F'_2)$.
Therefore, $w(F'_2) > w(F_2)$, a contradiction to the maximality of $F_2$ as a maximum-weight $[f, g]$-factor of $G_2$.
This completes the proof.
\end{proof}

In the remainder of this paper, we fix an input graph $G$ of PFTP for discussion.
Since $G$ is fixed, we abbreviate $V(G)$ and $E(G)$ as $V$ and $E$, respectively.
In addition to the notation introduced in Section~\ref{sec2}, we define the following for $G$.

\begin{notation}
\label{nota01}
Let $V_{\rm r}$ and $V_{\rm b}$ denote the sets of red and blue vertices in $G$, respectively.
By Problem~\ref{prob01}, $|V_{\rm r}| = r$ and $|V_{\rm b}| = 3n-r$.
Let $R$, $B$, and $X$ denote the sets of red, blue, and bichromatic edges in $G$, respectively.
For brevity, let $G_{\rm r} = G[R]$, $G_{\rm b} = G[B]$, and $G_{\rm x} = G[X]$.
\end{notation}

Now, we are ready to show the algorithm.
Roughly speaking, the algorithm computes two perfect fair-triangle packings $\mc{T}_0$ and $\mc{T}_1$, 
and outputs the one with greater weight.

We compute the first perfect fair-triangle packing $\mc{T}_0$ as follows.
Recall Notation~\ref{nota01}.
Since $|V_{\rm r}| = r$, $2(r-n) \le r$, and $G$ is complete, 
there exists a red matching of size $r-n$ in $G_{\rm r}$. 
We first use Lemma~\ref{lemma01} to compute a maximum-weight size-$(r-n)$ red matching $M_{\rm r}$ in $G_{\rm r}$ in $O(n^3)$ time.
Similarly, since $|V_{\rm b}| = 3n-r$, $2(2n-r) \le 3n-r$, and $G$ is complete, we can then use 
Lemma~\ref{lemma01} to compute a maximum-weight size-$(2n-r)$ blue matching $M_{\rm b}$ in $G_{\rm b}$ in $O(n^3)$ time. 
Obviously, $|V_{\rm r} \setminus V(M_{\rm r})| = 2n-r$ and $|V_{\rm b} \setminus V(M_{\rm b})| = r-n$.
Finally, we arbitrarily pair each edge in $M_{\rm r}$ with a distinct vertex in $V_{\rm b} \setminus V(M_{\rm b})$
to obtain $r-n$ fair triangles, 
and pair each edge in $M_{\rm b}$ with a distinct vertex in $V_{\rm r} \setminus V(M_{\rm r})$
to obtain $2n-r$ fair triangles. 
This yields the first perfect fair-triangle packing $\mc{T}_0$.

Similarly, we compute the second perfect fair-triangle packing $\mc{T}_1$ as follows.
Recall Notation~\ref{nota01} again. 
Since $n \le r = |V_{\rm r}| \le |V_{\rm b}|$ and $G$ is complete, $G_{\rm x}$ has a bichromatic matching of size $r \ge n$.
We use Lemma~\ref{lemma01} to compute a maximum-weight size-$n$ bichromatic matching $M_{\rm x}$ in $G_{\rm x}$ in $O(n^3)$ time.
Finally, we arbitrarily pair each edge in $M_{\rm x}$ with a distinct vertex in $V \setminus V(M_{\rm x})$ to obtain $n$ fair triangles. 
This yields the second perfect fair-triangle packing $\mc{T}_1$.

We summarize the algorithm in Figure~\ref{fig03}.

\begin{figure}[htb]
\begin{center}
\framebox{
\begin{minipage}{5.5in}
Algorithm {\sc Approx1}:\\
Input: A bicolored, edge-weighted complete graph $G = (V, E)$ with $|V|=3n$;
\begin{itemize}
\parskip=0pt
\item[1.]
	Use Lemma~\ref{lemma01} to compute a maximum-weight size-$(r-n)$ red matching $M_{\rm r}$ in $G_{\rm r}$ in $O(n^3)$ time.
\item[2.]
	Use Lemma~\ref{lemma01} to compute a maximum-weight size-$(2n-r)$ blue matching $M_{\rm b}$ in $G_{\rm b}$ in $O(n^3)$ time.
\item[3.]
	Arbitrarily pair each edge in $M_{\rm r}$ with a distinct vertex in $V_{\rm b} \setminus V(M_{\rm b})$  
	to obtain $r-n$ fair triangles, and pair each edge in $M_{\rm b}$ with a distinct vertex in $V_{\rm r} \setminus V(M_{\rm r})$  
	to obtain $2n-r$ fair triangles.  
	These $n$ triangles form the first perfect fair-triangle packing $\mc{T}_0$.
\item[4.]
	Use Lemma~\ref{lemma01} to compute a maximum-weight size-$n$ bichromatic matching $M_{\rm x}$ in $G_{\rm x}$ in $O(n^3)$ time.
\item[5.]
	Arbitrarily pair each edge in $M_{\rm x}$ with a distinct vertex in $V \setminus V(M_{\rm x})$ to obtain $n$ fair triangles.
	These $n$ triangles form the second perfect fair-triangle packing $\mc{T}_1$.
\item[6.]
          Return $\mc{T}_0$ if $w(\mc{T}_0) \ge w(\mc{T}_1)$; otherwise, return $w(\mc{T}_1)$. 
\end{itemize}
\end{minipage}}
\end{center}
\caption{A simple approximation algorithm for PFTP.} \label{fig03}
\end{figure}

We introduce the following notations for ease of presentation.

\begin{notation}
\label{nota02}
Fix an optimal perfect fair-triangle packing $\mc{T}^*$ of $G$. 
Let 
$R^* = R \cap E(\mc{T}^*)$,  $B^* = B \cap E(\mc{T}^*)$, and  $X^* = X \cap E(\mc{T}^*)$.
Clearly, $w(\mc{T}^*) = w(R^*)+w(B^*)+w(X^*)$.
\end{notation}

\begin{theorem}
\label{thm02}
{\sc Approx1} in Figure~\ref{fig03} is an $O(n^3)$-time $\frac 13$-approximation algorithm for PFTP.
\end{theorem}
\begin{proof}
Obviously, {\sc Approx1} runs in $O(n^3)$ time. It remains to analyze the approximation ratio.
By Observation~\ref{obs:1}, $R^*$ is a matching of size $r-n$ in $G_{\rm r}$ and $B^*$ is a matching of size $2n-r$ in $G_{\rm b}$. 
By the constructions of $M_{\rm r}$ and $M_{\rm b}$, $w(M_{\rm r}) \ge w(R^*)$ and $w(M_{\rm b}) \ge w(B^*)$. Hence, $w(\mc{T}_0) \ge w(R^*)+w(B^*)$.

By Observation~\ref{obs:1} again, $|X^*| = 2n$ and the edges in $X^*$ form $n$ vertex-disjoint $2$-paths.
For each of the $2$-paths, we select the heavier of its two edges. 
The selected edges form a bichromatic matching of size $n$ in $G_{\rm x}$ whose weight is at least $ \frac 12 w(X^*)$.
Consequently, by the construction of $M_{\rm x}$, $w(\mc{T}_1) \ge w(M_{\rm x}) \ge \frac 12 w(X^*)$.
Finally, we have 
\[
\max \{ w(\mc{T}_0), w(\mc{T}_1) \} \ge \frac 13 w(\mc{T}_0) + \frac 23 w(\mc{T}_1) 
\ge \frac 13 (w(R^*)+w(B^*)+w(X^*)) = \frac 13 w(\mc{T}^*).
\]
This completes the proof.
\end{proof}

\section{A randomized $(\frac {16}{47}-\epsilon)$-approximation algorithm for PFTP}\label{sec:random}

Let $\epsilon \in (0, 1]$ be a fixed sufficiently small constant.
For simplicity, we assume that $\frac 1\epsilon$ is an integer.
In this section, we present a randomized $(\frac {16}{47}-\epsilon)$-approximation algorithm for PFTP. 
Roughly speaking, the algorithm computes five perfect fair-triangle packings $\mc{T}_0, \ldots, \mc{T}_4$, 
and outputs the one with the greatest weight. $\mc{T}_0$ is computed by {\sc Approx1}. 
In the remainder of this section, we show how to compute $\mc{T}_1, \ldots, \mc{T}_4$ using an approach completely different from that of {\sc Approx1}. 

Recall Notations~\ref{nota01} and~\ref{nota02}.
Also recall that a {\em maximum-weight $[1, 2]$-factor} (cf. Lemma~\ref{lemma01} or~\cite{Gab83}) in $G_{\rm x}$ is a set  $F \subseteq X$
such that the degree of each vertex in $G[F]$ is $1$ or $2$.
ByObservation~\ref{obs:1}, $X^*$ is a $[1, 2]$-factor in $G_{\rm x}$.
Therefore, our algorithms always start by performing the following step:
\begin{itemize}
\item Compute a maximum-weight $[1, 2]$-factor $F$ in $G_{\rm x}$ in $O(n^3)$ time~\cite{Gab83}.
Let $\mc{F} = G[F]$. 
\end{itemize}
Clearly, $w(\mc{F}) \ge w(X^*)$.
Recall that a single vertex is not regarded as a path.
Since $F$ is a $[1, 2]$-factor, each vertex in $\mc{F}$ is of degree $1$ or $2$.
This implies that each connected component of $\mc{F}$ is either a cycle or a path.
A cycle or a path in $\mc{F}$ is {\em long} if it has at least $\frac 2\epsilon$ edges; otherwise, it is {\em short}.

\subsection{Preprocessing long connected components of $\mc{F}$}\label{subsec:prep}
Once our algorithms obtain $\mc{F}$, they decompose each long connected component $C$ of $\mc{F}$ 
into a collection of short paths, losing at most an $\epsilon$ fraction of the total weight, as follows.
Let $c = |E(C)| \ge \frac 2\epsilon$, and $e_1, \ldots, e_c$ denote the edges of $C$ in the order in which they appear along $C$.

First consider the case where $C$ is a cycle. Let $\ell$ be the largest integer such that $\frac \ell \epsilon \le c$. 
Moreover, for each $i \in \{ 1, 2, \ldots, \frac 1\epsilon\}$, let $E_i = \{ e_i, e_{i+1/\epsilon}, \ldots, e_{i+(\ell-1)/\epsilon} \}$. 
Clearly, $E_1$, \ldots, $E_{1/\epsilon}$ form a partition of $\{e_1, e_2, \ldots, e_{\ell/\epsilon}\}$.
By the choice of $\ell$, we have $\frac {\ell+1} \epsilon > c$, and hence $c-\frac \ell \epsilon \le \frac 1\epsilon-1$.
For each $i \in \{ 1, 2, \ldots, \frac 1\epsilon \}$, removing the edges in $E_i$ from $C$ yields a collection $C_i$ of vertex-disjoint paths.
Among the paths in $C_i$, since $c-\frac \ell \epsilon \le \frac 1\epsilon-1$, at most one path has length in the interval 
$[\frac 1\epsilon-1, \frac 2\epsilon-2]$, while each of the remaining paths has length exactly $\frac 1\epsilon-1$. 
So, each path in $C_i$ is short. Consider an integer $i\in\{1,2,\ldots,1/\epsilon\}$ such that $w(E_i)$ is minimized. 
Since $\sum_{j=1}^{1/\epsilon}w(E_j) \le w(C)$, $w(E_i) \le \epsilon\cdot w(C)$. 
We now remove the edges in $E_i$ from $C$. This removal decreases $w(C)$ by at most an $\epsilon$ fraction. 

Next consider the case where $C$ is a path. Let $\ell$ be the largest integer such that $\frac \ell \epsilon + 2 \le c$.
Moreover, for each $i \in \{ 2, 3, \ldots, \frac 1\epsilon+1\}$, let $E_i = \{ e_i, e_{i+1/\epsilon}, \ldots, e_{i+(\ell-1)/\epsilon} \}$. 
Clearly, $E_2$, \ldots, $E_{1/\epsilon+1}$ form a partition of $\{e_2, e_3, \ldots, e_{\ell/\epsilon+1}\}$.
By the choice of $\ell$, we have $\frac \ell \epsilon+2 \le c$ and $\frac {\ell+1} \epsilon+2 > c$, and hence
 $2 \le c-\frac \ell \epsilon \le \frac 1\epsilon+1$. 
For each $i \in \{ 2, 3, \ldots, \frac 1\epsilon+1\}$, removing the edges in $E_i$ from $C$ yields a collection $C_i$ of vertex-disjoint paths.
The path containing $e_1$ in $C_i$ has length $i-1 \in [1, \frac 1\epsilon]$, the path containing $e_c$ in $C_i$ has length 
$c-i-\frac {\ell-1}\epsilon \in [1, \frac 2\epsilon-1]$, and each of the remaining path has length exactly $\frac 1\epsilon-1$.
So, each path in $C_i$ is short. Consider an integer $i\in\{2,3,\ldots,1/\epsilon+1\}$ such that $w(E_i)$ is minimized. 
Since $\sum_{i=2}^{1/\epsilon+1} w(E_i) \le w(C)$, $w(E_i) \le \epsilon\cdot w(C)$. 
We now remove the edges in $E_i$ from $C$. This removal decreases $w(C)$ by at most an $\epsilon$ fraction. 

\begin{remark}
\label{remark01}
Suppose that we have modified $\mc{F}$ by decomposing each long connected component $C$ of $\mc{F}$ as described above. 
Then, $\mc{F}$ has the following properties:
\begin{itemize}
\item[P1.] Each connected component $K$ of $\mc{F}$ is a short cycle or a short path. Moreover, 
if $K$ is a cycle, then it has even length (because each edge in $\mc{F}$ is bichromatic).

\item[P2.] Let $\mc{C}$ and $\mc{P}$ denote the sets of cycles and paths in $\mc{F}$, respectively.
Then, $w(\mc{C})+w(\mc{P}) = w(\mc{F}) \ge (1-\epsilon)w(X^*)$.
\end{itemize}
\end{remark}

Consider a bichromatic edge $\{ u, v \}$ in $G$.
Without loss of generality, assume that $u$ is red and $v$ is blue. 
Then, $\{ u, v \}$ must belong to one of the following five types:
\begin{itemize}
\item {\em Type-1}: $u$ and $v$ belong to the same connected component of $\mc{F}$.

\item {\em Type-2}: $u$ and $v$ belong to different cycle components of $\mc{F}$.

\item {\em Type-3}: $u$ belongs to a cycle component of $\mc{F}$ and $v$ belongs to a path component of $\mc{F}$.

\item {\em Type-4}: $u$ belongs to a path component of $\mc{F}$ and $v$ belongs to a cycle component of $\mc{F}$.

\item {\em Type-5}: $u$ and $v$ belong to different path components of $\mc{F}$.
\end{itemize}
\begin{notation}
\label{nota03}
For each $i \in \{ 1, \ldots, 5 \}$, let $X_i \subseteq X$ be the set of bichromatic edges of Type-$i$, and let $X^*_i = X_i \cap E(X^*)$.
In other words, $X^*_i$ consists of all Type-$i$ bichromatic edges in the optimal solution.
Clearly, $w(X^*) = \sum_{i=1}^5 w(X^*_i)$.
See Figure~\ref{fig04} for an example.
\end{notation}

\begin{lemma}
\label{lemma02}
For each $i \in \{ 1, \ldots, 5 \}$, a maximum-weight matching $M_i$ in $G[X_i]$ satisfies $w(M_i) \ge \frac 12 w(X^*_i)$.
\end{lemma}
\begin{proof}
By Observation~\ref{obs:1}, each connected component of $G[X^*]$ is a $2$-path. 
Consider an $i \in \{ 1, \ldots, 5 \}$. 
Clearly, each connected component of $G[X_i^*]$ is either a single vertex, a $1$-path, or a $2$-path. 
Thus, $X_i^*$ can be partitioned into two matchings $N_1$ and $N_2$. 
Obviously, both $N_1$ and $N_2$ are mathings in $G[X^*_i] \subseteq G[X_i]$, and $\max\{w(N_1),w(N_2)\} \ge \frac 12 w(X^*_i)$. 
Hence, $w(M_i) \ge \frac 12 w(X^*_i)$.
\end{proof}

\begin{figure}[thb]
\begin{center}
\begin{tikzpicture}[scale=0.45,transform shape]
         
\foreach \x/\y in { 0/0, 0/3.5 } {
         \draw[thick, line width = 1pt] (\x, \y) -- (\x+1.5, \y);
         \draw[thick, line width = 1pt] (\x, \y) -- (\x, \y-1.5);
         \draw[thick, line width = 1pt] (\x+1.5, \y) -- (\x+1.5, \y-1.5);
         \draw[thick, line width = 1pt] (\x, \y-1.5) -- (\x+1.5, \y-1.5);
         \fill[red] (\x, \y) circle(.15);
         \fill[blue] (\x+1.5, \y) circle(.15);
         \fill[red] (\x+1.5, \y-1.5) circle(.15);
         \fill[blue] (\x, \y-1.5) circle(.15);
}

\foreach \x/\y in { 6/3, 6/1, 6/-1 } {
         \draw[thick, line width = 1pt] (\x, \y) -- (\x+1.5, \y);
         \draw[thick, line width = 1pt] (\x+1.5, \y) -- (\x+3, \y);
         \fill[blue] (\x, \y) circle(.15);
         \fill[red] (\x+1.5, \y) circle(.15);
         \fill[blue] (\x+3, \y) circle(.15);
}

\draw[thick, line width = 1pt] (4.5, -1) -- (6, -1);
\fill[red] (4.5, -1) circle(.15);
\draw[thick, line width = 1pt] (9, -1) -- (10.5, -1);
\fill[red] (10.5, -1) circle(.15);

\draw[densely dashed, line width = 1pt] (4.5, -1) -- (1.5, 0);
\draw[densely dashed, line width = 1pt] (6, 1) -- (1.5, 2);
\draw[densely dashed, line width = 1pt] (7.5, 3) -- (9, 1);
\draw[densely dashed, line width = 1pt] (0, 0) -- (0, 2);
\draw [densely dashed, line width = 1pt] (6, -1) to [out = -20, in = -160] (10.5, -1);

\node[font=\fontsize{24}{6}\selectfont] at (8, -2) {$e_1$};
\node[font=\fontsize{24}{6}\selectfont] at (-0.7, 1) {$e_2$};
\node[font=\fontsize{24}{6}\selectfont] at (3, -1) {$e_4$};
\node[font=\fontsize{24}{6}\selectfont] at (3.3, 2) {$e_3$};
\node[font=\fontsize{24}{6}\selectfont] at (7.5, 2) {$e_5$};
\end{tikzpicture}
\end{center}
\caption{The different types of edges in $X^*$.
In the graph, $\mc{F}$ consists of two $4$-cycles, two $2$-paths, and one $4$-path.
For each $i \in \{ 1, \ldots, 5 \}$, the dashed edge $e_i$ is in $X^*\setminus E(\mc{F})$ 
and is of Type-$i$.
\label{fig04}}
\end{figure}
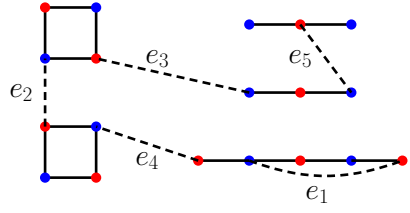

\subsection{Computing $\mc{T}_1$}

In this subsection, we design an algorithm that computes a perfect fair-triangle packing $\mc{T}_1$ satisfying $w(\mc{T}_1) \ge w(X^*_1)$.
A $2$-path is {\em red-dominant} (respectively, {\em blue-dominant}) if its vertices, in order, are colored red, blue, and red 
(respectively, blue, red, and blue).
A {\em fair $(1,2)$-path packing} is a subgraph of $G[X_1]$ in which each connected component is a bichromatic edge, a red-dominant path, or a blue-dominant path.
Since each edge in a fair $(1,2)$-path packing is in $X_1$, it must connect two vertices in the same connected component of $\mc{F}$.
Moreover, a fair $(1,2)$-path packing can be characterized by a triple $(i, j, k)$, where $i$, $j$, and $k$ denote the numbers of bichromatic edges, 
red-dominant paths, and blue-dominant paths in the packing, respectively. 
For brevity, we refer to such a fair $(1,2)$-path packing as an {\em $(i,j,k)$-packing}.

Intuitively, we will compute $\mc{T}_1$ by first computing a maximum-weight fair $(1,2)$-path packing in $G[X_1]$ and then 
transforming it to a perfect fair-triangle packing in $G$.

Let $C_1, \ldots, C_q$ denote the connected components of $\mc{F}$, where $q$ is the number of connected components of $\mc{F}$.
By P1 in Remark~\ref{remark01}, for each $1 \le h \le q$, $C_h$ is short and hence $|V(C_h)| \le \frac 2\epsilon$.
Let $C'_h$ denote the subgraph induced by $V(C_h)$.
Clearly, $C'_h$ is a subgraph of $G[X_1]$ and every fair $(1,2)$-path packing in $C'_h$ is an $(i, j, k)$-packing with $i+ j+k \le \frac{1}{2}|V(C_h)|$.
For each triple $(i,j,k)$ with $i+ j+k \le \frac{1}{2}|V(C_h)|$, let $\mc{Q}_h(i, j, k)$ denote a maximum-weight $(i, j, k)$-packing in $C'_h$.
If no $(i,j,k)$-packing exists in $C'_h$, we define $w(\mc{Q}_h(i, j, k)) = -\infty$ and regard $\mc{Q}_h(i, j, k)$ as {\em invalid}.
Since $|V(C_h)| \le \frac 2\epsilon$, computing $\mc{Q}_h(i, j, k)$ for each triple $(i, j, k)$ takes $O(1)$ time by brute force. 
Consequently, computing $\mc{Q}_h(i, j, k)$ for all such triples takes $O(1)$ total time.

For each $h \in\{ 1, 2, \ldots, q\}$, let $\mc{F}'_h$ denote the union of $C'_1, \ldots, C'_h$. 
For each triple $(i,j,k)$ with $i+j+k \le n$, let $\mc{Q}'_h(i, j, k)$ denote a maximum-weight $(i, j, k)$-packing in $\mc{F}'_h$.
If no $(i,j,k)$-packing exists in $\mc{F}'_h$, we define $w(\mc{Q}'_h(i, j, k)) = -\infty$ and regard $\mc{Q}'_h(i, j, k)$ as {\em invalid}.
We claim that computing $\mc{Q}'_h(i, j, k)$ for all triples $(i,j,k)$ with $i+j+k \le n$ can be done in $O(n^3h)$ total time 
via dynamic programming. 
To see the claim, first note that for each triple $(i,j,k)$ with $i+j+k \le n$, 
$\mc{Q}'_1(i, j, k) = \mc{Q}_1(i, j, k)$ if $i+ j+k \le \frac{1}{2}|V(C_h)|$, while $\mc{Q}'_1(i, j, k)$ is invalid otherwise.
Hence, it takes $O(n^3)$ total time to compute $\mc{Q}'_1(i, j, k)$ for all triples $(i,j,k)$ with $i+j+k \le n$.
Now, consider an arbitrary $h \in \{2, 3, \ldots, q\}$.
Suppose that we have computed $\mc{Q}'_{h-1}(i, j, k)$ for all triples $(i,j,k)$ with $i+j+k \le n$ in $O(n^3(h-1))$ time. 
To compute $\mc{Q}'_{h}(i, j, k)$ for a triple $(i,j,k)$ with $i+j+k \le n$, we first find 
a triple $(i',j',k')$ such that $i'\le i$, $j'\le j$, $k'\le k$, $i'+j'+k'\le \frac 12|V(C_h)|$, 
and $w(\mc{Q}_h(i', j', k')) + w(\mc{Q}'_{h-1}(i-i', j-j', k-k'))$ is maximized.
Clearly, if neither $w(\mc{Q}_h(i', j', k'))$ nor $w(\mc{Q}'_{h-1}(i-i', j-j', k-k'))$ is $-\infty$, then 
$\mc{Q}'_h(i, j, k)$ is the union of $\mc{Q}_h(i', j', k')$ and $\mc{Q}'_{h-1}(i-i', j-j', k-k')$; 
otherwise, $\mc{Q}'_h(i, j, k)$ is invalid. 
Since computing $\mc{Q}_h(i',j',k')$ for all required $(i', j', k')$ takes $O(1)$ time, so does $\mc{Q}'_h(i, j, k)$. 
Thus, it takes $O(n^3h)$ total time to compute $\mc{Q}'_h(i, j, k)$ for all triples $(i,j,k)$ with $i+j+k \le n$. 
This finishes the proof of the claim. Since $q \le \frac 32n$, the claim implies that 
it takes $O(n^4)$ total time to compute $\mc{Q}'_q(i, j, k)$ for all triples $(i,j,k)$ with $i+j+k \le n$.

Unfortunately, even if $\mc{Q}'_q(i, j, k)$ is valid, it may be impossible to transform it into 
a perfect fair-triangle packing in $G$. So, we need the next definition. 

\begin{definition}
\label{def01}
An $(i,j,k)$-packing in $\mc{F}'_q$ is {\em feasible} if $j \le r-n$, $k \le 2n-r$, and $i+j+k \le n$, 
where $r = |V_{\rm r}|$ and $3n = |V|$.
\end{definition}

\begin{lemma}
\label{lemma03}
Given a feasible $(i,j,k)$-packing $\mc{K}$ in $\mc{F}'_q$, we can transform it into a perfect fair-triangle packing in $G$ 
in $O(n)$ time without losing any weight.
\end{lemma}
\begin{proof}
Obviously, $V \setminus V(\mc{K})$ contains exactly $r-i-2j-k$ red vertices 
and exactly $3n-r-i-j-2k$ blue vertices.
By Definition~\ref{def01}, one can verify that 
\[
0 \le n-i-j-k \le \min \{ r-i-2j-k, 3n-r-i-j-2k \}.
\]
So, $V \setminus V(\mc{K})$ has at least $n-i-j-k$ red vertices and at least $n-i-j-k$ blue vertices. 
Since $G$ is complete, we can select $n-i-j-k$ vertex-disjoint bichromatic edges from the subgraph induced by $V \setminus V(\mc{K})$ 
and obtain an $(n-j-k, j, k)$-packing $\mc{K}'$ in $G$ by adding the selected edges to $\mc{K}$. 

We next transform $\mc{K}'$ into a perfect fair-triangle packing in $G$ as follows. 
First, we complete each $2$-path in $\mc{K}'$ into a triangle by adding the edge joining its endpoints in $G$. 
Note that $V \setminus V(\mc{K}')$ has $n-j-k$ vertices. 
We then pair each vertex of  $V \setminus V(\mc{K}')$ with a distinct 
edge component of $\mc{K}'$, and complete each pair into a triangle in $G$. This completes the transformation 
of $\mc{K}'$ into a perfect fair-triangle packing in $G$ with weight at least $w(\mc{K})$. 
\end{proof}

Recall that we have computed $\mc{Q}'_q(i, j, k)$ for each triple $(i,j,k)$ with $i+j+k \le n$. 
Among these packings, we select a feasible one of maximum weight and transform it into 
$\mc{T}_1$ as described in Lemma~\ref{lemma03}. The whole algorithm is summarized in Figure~\ref{fig05}.

\begin{figure}[htb]
\begin{center}
\framebox{
\begin{minipage}{5.5in}
{\em Algorithm for computing $\mc{T}_1$}:\\
Input: $G$ and $\mc{F}$;
\begin{itemize}
\parskip=0pt
\item[1.]
	Let $C_1, \ldots, C_q$ denote the connected components of $\mc{F}$.
	Fof each $1 \le h \le q$, let $C'_h$ be the subgraph of $G[X_1]$ induced by $V(C_h)$.
	For each $1 \le h \le q$ and for each triple $(i,j,k)$ satisfying $i+ j+k \le \frac{1}{2}|V(C_h)|$, 
         compute a maximum-weight $(i,j,k)$-packing $\mc{Q}_h(i, j, k)$ in $C'_h$. 
	({\em Comment:} This step takes $O(q)$ time.)
         
\item[2.]
         For each $1 \le h \le q$, let $\mc{F}'_h$ denote the union of $C'_1, \ldots, C'_h$.
	For each triple $(i,j,k)$ with $i+j+k \le n$, compute a maximum-weight $(i,j,k)$-packing $\mc{Q}'_q(i, j, k)$ in $\mc{F}'_q$ 
	if one exists, using dynamic programming. ({\em Comment:} This step takes $O(n^4)$ time.)

\item[3.]
	Among the packings computed in Step~2, select a feasible one of maximum weight. 
         
\item[4.]
	Transform the packing selected in Step~3  into a perfect fair-triangle packing in $G$, as described in Lemma~\ref{lemma03}. 
         
\item[5.] Return the packing obtained in Step~4 as $\mc{T}_1$.
\end{itemize}
\end{minipage}}
\end{center}
\caption{Dynamic-programming algorithm for computing $\mc{T}_1$.} \label{fig05}
\end{figure}

\begin{lemma}
\label{lemma04}
The algorithm in Figure~\ref{fig05} computes a perfect fair-triangle packing $\mc{T}_1$ in $O(n^4)$ time satisfying 
$w(\mc{T}_1) \ge w(X^*_1)$.
\end{lemma}
\begin{proof}
The time complexity is clear from the discussion above. 
Let $\mc{K}$ denote the feasible packing found in Step~3 in Figure~\ref{fig05}.
By Lemma~\ref{lemma03}, $w(\mc{T}_1) \ge w(\mc{K})$.
So, it remains to show that $w(\mc{K}) \ge w(X^*_1)$.

Consider the subgraph $G^*_1=(V(X^*_1), X^*_1)$ of $G$. 
Clearly, $G^*_1$ is a subgraph of $G[X_1]$.
By Observation~\ref{obs:1}, each connected component of $G^*_1$ is a bichromatic edge, a red-dominant $2$-path, or a blue-dominant $2$-path. 
Moreover, since $\mc{T}^*$ has exactly $n$ triangles and each connected component of $G^*_1$ is contained in  
a distinct triangle of $\mc{T}^*$, 
$G^*_1$ is an $(i,j,k)$-packing for some triple $(i,j,k)$ with $i+j+k \le n$. 
Furthermore, since each red-dominant (respectively, blue-dominant) $2$-path in $G^*_1$ is contained in  
a distinct red-dominant (respectively, blue-dominant) triangle of $\mc{T}^*$,
we also have $j \le r-n$ and $k \le 2n-r$. 
Clearly, each edge in $X^*_1$ must belong to some $C'_h$.
In other words, $G^*_1$ is a subgraph of $\mc{F}'_q$. 
Therefore, by Definition~\ref{def01}, $G^*_1$ is feasible and $w(\mc{K}) \ge w(X^*_1)$.
\end{proof}

\subsection{Two key lemmas}

In this section, we prove two key lemmas showing how to construct a perfect fair-triangle packing from a specific collection of vertex-disjoint paths.

\begin{definition}
\label{def02}
An {\em alternating path} (respectively, {\em cycle}) in $G$ is a path (respectively, cycle) whose vertices alternate in color.
Moreover, an alternating path of odd (respectively, even) length is defined as an {\em odd} (respectively, {\em even}) alternating path. 
\end{definition}
By P1 in Remark~\ref{remark01}, every path or cycle component of $\mc{F}$ is alternating. 
In Figure~\ref{fig04}, every path component of $\mc{F}$ is an even alternating path. 

\begin{definition}
\label{def03}
A path in $G$ is {\em special} if it is obtained by extending an odd alternating path at its blue endpoint with a blue edge.
For example, in Figure~\ref{fig13}, deleting the dashed edge $e_4$ results in a graph containing a unique solid $2$-path, 
which is special.
\end{definition}

The following two lemmas provide the foundation for constructing a perfect fair-triangle packing from a collection of vertex-disjoint paths:
Lemma~\ref{lemma05} applies when all the paths are alternating, whereas Lemma~\ref{lemma06} applies when 
each path is an odd alternating path or a special path.

\begin{lemma}
\label{lemma05}
Suppose that $\mc{K}$ is a collection of vertex-disjoint alternating paths with $V(\mc{K}) = V$.
Then, in $O(n)$ time, we can construct a perfect fair-triangle packing with total weight at least $\frac 23 w(\mc{K})$.
\end{lemma}
\begin{proof}
First, we form a single cycle $C$ from $\mc{K}$ by adding edges between the endpoints of distinct paths in $\mc{K}$.
This is possible because $G$ is complete.
Clearly, $w(C) \ge w(\mc{K})$.
We claim that for any three consecutive vertices $v_1, v_2, v_3$ on $C$, at least one is red and at least one is blue. 
To see the claim, first note that since every path has length at least~1, the construction of $C$ guarantees that
 $\{v_1,v_2\}$ or $\{v_2,v_3\}$ is an edge of some $P \in \mc{K}$. 
Since $P$ is alternating, all its edges are bichromatic. 
So, the claim holds. 

We next show how to constrcut a perfect fair-triangle packing from $C$.
Since $V(\mc{K}) = V$, $|E(C)| = 3n$. 
Let $e_1, \ldots, e_{3n}$ be the edges of $C$ in order. 
Clearly, $E(C)$ can be partitioned into three disjoint matchings $\{ e_i, e_{i+3}, \ldots, e_{i+3(n-1)} \}$ for $i = 1, 2, 3$. 
Among these three matchings, we select one of minimum weight and delete its edges from $C$. 
Thus, $C$ loses at most one-third of its weight and becomes a collection of vertex-disjoint $2$-paths. 
We then complete each $2$-path in $C$ into a triangle by adding the edge joining its endpoints in $G$. 
Each resulting triangle is fair by the claim proved above. 
This completes the proof.
\end{proof}

\begin{lemma}
\label{lemma06}
Suppose that $\mc{K}$ is a collection of vertex-disjoint paths with $V(\mc{K}) = V$,
each of which is either an odd alternating path or a special path.
Then, in $O(n)$ time, we can construct a perfect fair-triangle packing with total weight at least $\frac 23 w(\mc{K})$.
\end{lemma}
\begin{proof}
By Definitions~\ref{def02} and~\ref{def03}, each path in $\mc{K}$ has endpoints of different colors. 
First, we form a single cycle $C$ from $\mc{K}$ by adding bichromatic edges between the endpoints of distinct paths in $\mc{K}$.
This is possible because $G$ is complete. 
Clearly, $w(C) \ge w(\mc{K})$.

The construction of $C$ alongside Definitions~\ref{def02} and~\ref{def03} implies that
every pair of consecutive edges on $C$ contains at least one bichromatic edge---which either connects distinct paths in $\mc{K}$ or lies on some path in $\mc{K}$. 
So, for any three consecutive vertices on $C$, at least one is red and at least one is blue. 
Now, we obtain a perfect fair-triangle packing from $C$ as in the the proof of Lemma~\ref{lemma05}. 
\end{proof}

\subsection{Randomly breaking cycles in $\mc{F}$}

Hereafter, {\em breaking a cycle} means deleting one or more edges from the cycle. 

Recall that $\mc{C}$ and $\mc{P}$ denote the sets of cycles and paths in $\mc{F}$, respectively (see P2 in Remark~\ref{remark01}). 
We will introduce additional notations throughout the remainder of this paper. 
Table~\ref{tab01} summarizes the most important notations for ease of reference. 

By P1 in Remark~\ref{remark01}, we can transform $\mc{F}$ into a collection $\mc{K}$ of vertex-disjoint alternating paths 
with $V(\mc{K}) = V$ by deleting an arbitrary edge from each cycle in $\mc{C}$. 
We can then use Lemma~\ref{lemma05} to construct a perfect fair-triangle packing $\mc{Q}$ from $\mc{K}$ 
satisfying $w(\mc{Q}) \ge \frac 23 w(\mc{K})$. 
However, since $\mc{F}$ may contain many $4$-cycles, 
it is possible that $w(\mc{K}) \le \frac 34 w(\mc{F})$, in which case we may obtain $w(\mc{Q}) \le \frac 12 w(\mc{F})$. 

Since simply deleting one arbitrary edge from each cycle in $\mc{C}$ is insufficient, 
we next present a randomized method that transforms $\mc{C}$ into a collection $\mc{C}_1$ of vertex-disjoint odd alternating paths. 
Crucially, we guarantee that $\mathbb{E}[w(\mc{C}_1)] \ge \frac 34 w(\mc{C})$ and that $\Pr[d_{\mc{C}_1}(v) = 1] = \frac 12$ for every $v \in V(\mc{C})$. 
Moreover, if an edge $\{ u, v \} \in X$ of Type~$2$, satisfies $d_{\mc{C}_1}(u)=d_{\mc{C}_1}(v)=1$, 
then it could be added to $\mc{C}_1$ while preserving the property that $\mc{C}_1$ is a collection of vertex-disjoint odd alternating paths. 
For example, if $\mc{F}$ is as shown in Figure~\ref{fig04} and we break its cycles as shown in Figure~\ref{fig06}, 
then the edge $e_2$ can be added back to $\mc{C}_1$, 
after which $\mc{C}_1$ consists of a single alternating $7$-path.

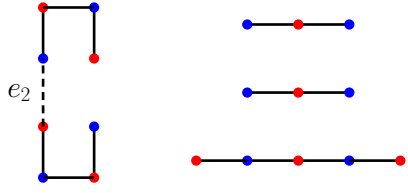
\begin{figure}[thb]
\begin{center}
\begin{tikzpicture}[scale=0.45,transform shape]
         
\foreach \x/\y in { 0/0, 0/3.5 } {
         \draw[thick, line width = 1pt] (\x, \y) -- (\x, \y-1.5);
         \draw[thick, line width = 1pt] (\x+1.5, \y) -- (\x+1.5, \y-1.5);
         \fill[red] (\x, \y) circle(.15);
         \fill[blue] (\x+1.5, \y) circle(.15);
         \fill[red] (\x+1.5, \y-1.5) circle(.15);
         \fill[blue] (\x, \y-1.5) circle(.15);
}

\draw[thick, line width = 1pt] (0, -1.5) -- (1.5, -1.5);
\draw[thick, line width = 1pt] (0, 3.5) -- (1.5, 3.5);

\foreach \x/\y in { 6/3, 6/1, 6/-1 } {
         \draw[thick, line width = 1pt] (\x, \y) -- (\x+1.5, \y);
         \draw[thick, line width = 1pt] (\x+1.5, \y) -- (\x+3, \y);
         \fill[blue] (\x, \y) circle(.15);
         \fill[red] (\x+1.5, \y) circle(.15);
         \fill[blue] (\x+3, \y) circle(.15);
}

\draw[thick, line width = 1pt] (4.5, -1) -- (6, -1);
\fill[red] (4.5, -1) circle(.15);
\draw[thick, line width = 1pt] (9, -1) -- (10.5, -1);
\fill[red] (10.5, -1) circle(.15);

\draw[densely dashed, line width = 1pt] (0, 0) -- (0, 2);

\node[font=\fontsize{24}{6}\selectfont] at (-0.7, 1) {$e_2$};
\end{tikzpicture}
\end{center}
\caption{Breaking the cycles in $\mc{F}$ shown in Figure~\ref{fig04}. 
Afterward, the edge $e_2$ can be used to connect two resulting alternating $3$-paths and form an alternating $7$-path.
\label{fig06}}
\end{figure}

We process the cycles $C$ in $\mc{C}$ independently at random as follows. 
Let $c = |E(C)|$. 
We select an edge $e_1 \in E(C)$ uniformly at random, and then 
let $e_1, \ldots, e_c$ denote the edges of $C$ in the order in which they appear along $C$.
By P1 in Remark~\ref{remark01}, $c$ is even. So, either $c \equiv 0 \pmod{4}$ or $c \equiv 2 \pmod{4}$.
In the former case, we delete $e_1, e_5, \ldots, e_{c-3}$ from $C$.
In the latter case, we delete $e_1, e_5, \ldots, e_{c-5}$ from $C$, and then delete $e_{c-1}$ with probability $\frac 12$.
It is easy to verify that these deletions always break $C$ into a set of vertex-disjoint odd alternating paths.
The whole process is summarized in Figure~\ref{fig07}.

\begin{figure}[htb]
\begin{center}
\framebox{
\begin{minipage}{5.5in}
{\em Algorithm for breaking a cycle in $\mc{C}$}:\\
Input: A cycle $C$ in $\mc{C}$;
\begin{itemize}
\parskip=0pt
\item[1.]
	Select an edge $e_1 \in E(C)$ uniformly at random.
\item[2.]
 	Let $c=|E(C)|$, and $e_1, \ldots, e_c$ be the edges of $C$ in the order in which they appear along $C$.     
\item[3.]
	If $c \equiv 0 \pmod{4}$, then delete $e_1, e_5, \ldots, e_{c-3}$ from $C$.
\item[4.]
	If $c \equiv 2 \pmod{4}$, then delete $e_1, e_5, \ldots, e_{c-5}$ from $C$, and then 
	delete $e_{c-1}$ with probability $\frac 12$.
\item[5.]
         Return the modified graph $C$, which is a collection of vertex-disjoint odd alternating paths.
\end{itemize}
\end{minipage}}
\end{center}
\caption{Breaking a cycle in $\mc{C}$ at random.} \label{fig07}
\end{figure}

Let $\mc{C}_1$ be the graph obtained from $\mc{C}$ by breaking the cycles in $\mc{C}$ as described above.
As noted above, each connected component of $\mc{C}_1$ is an odd alternating path.

\begin{lemma}
\label{lemma07}
Let $C$ be a cycle in $\mc{C}$.
For each  $e \in E(C)$, $\Pr[e \in E(\mc{C}_1)] = \frac 34$.
Moreover, for each  $u \in V(C)$, $\Pr[d_{\mc{C}_1}(u) = 1] = \frac 12$.
\end{lemma}
\begin{proof}
Let $c=|E(C)|$ and $c'$ be the number of edges deleted from $C$ in Steps~3 and~4 of Figure~\ref{fig07}.
If $c \equiv 0 \pmod{4}$, then $c' = \frac c4$; otherwise, either 
$c' = \frac {c-2}4$ or $c'=\frac {c+2}4$, and either case occurs with probability $\frac 12$.
Therefore, $\mathbb{E}[c'] = \frac c4$. 
By Step~1 in Figure~\ref{fig07}, each edge $e \in E({C})$ is equally likely to be deleted, and hence $\Pr[e \notin E(\mc{C}_1)] = \frac 14$.
Thus, the first assertion in the lemma holds.

Consider a $u\in V(C)$. 
Let $e_1$ and $e_2$ be the two edges incident to $u$ in $C$.
By the steps in Figure~\ref{fig07}, $|\{e_1,e_2\} \cap E(\mc{C}_1)| \le 1$.
Hence, $d_{\mc{C}_1}(u) = 1$ if $e_1\notin E(\mc{C}_1)$ or $e_2\notin E(\mc{C}_1)$. 
Now, since $\Pr[e_i \notin E(\mc{C}_1)] = \frac 14$ for each $i \in \{ 1, 2 \}$, 
the second assertion in the lemma holds.
\end{proof}

For each pair $(u,t)$ of vertices in the same cycle in $\mc{C}$, 
let $A_1(u, t)$ denote the event that $u$ and $t$ are the endpoints of the same path in $\mc{C}_1$, 
and $A_2(u, t)$ denote the event that $u$ and $t$ are the endpoints of two different paths in $\mc{C}_1$.

\begin{lemma}
\label{lemma08}
{\em \cite{TC10}}
For each pair $(u,t)$ of vertices in the same cycle $C$ in $\mc{C}$, 
we have $\frac 12 \Pr[A_1(u,t)] + \frac 14 \Pr[A_2(u,t)] \le \frac 18$.
\end{lemma}
\begin{proof}
When $A_2(u, t)$ occurs, $d_{\mc{C}_1}(u) = 1$. 
So, $\Pr[A_2(u,t)] \le \frac 12$ by Lemma~\ref{lemma07}.
Thus, it suffices to consider those pairs $(u,t)$ with $\Pr[A_1(u,t)] > 0$. 
Fix such a pair $(u,t)$. 
Then, by Step~5 in Figure~\ref{fig07}, $u$ and $t$ have different colors.

Since $C$ is a cycle of even length $c \ge 4$, $C$ contains two internally vertex-disjoint paths between $u$ and $t$. 
Let $\ell$ be the length of the shorter path.
Since $u$ and $t$ have different colors, $\ell$ is odd.
We distinguish four cases as follows. 

{\em Case 1:} $c = 4$. 
In this case, the algorithm in Figure~\ref{fig07} deletes a single edge from $C$. 
So, $\Pr[A_2(u,t)] = 0$. 
Moreover, $A_1(u,t)$ occurs only if $\ell=1$ and the edge $\{ u, t \}$ is deleted from $C$ by the algorithm.
Therefore, $\Pr[A_1(u,t)] = \frac 14$ and we are done.

{\em Case 2:} $c = 6$. 
In this case, $\ell \in \{ 1, 2, 3 \}$. 
Since $\ell$ is odd, we have $\ell \in \{ 1, 3 \}$. 
In Figure~\ref{fig08}, we enumerate all possible cases when $d_{\mc{C}_1}(u)=d_{\mc{C}_1}(t)=1$. 
If $\ell=1$, then $A_1(u,t)$ occurs when $C$ is broken as shown in the frist two graphs in Figure~\ref{fig08}, 
whereas $A_2(u,t)$ occurs when $C$ is broken as shown in the third and fourth graphs in Figure~\ref{fig08}; 
thus, $\Pr[A_1(u,t)]=\Pr[A_2(u,t)]=\frac 16$. 
Moreover, if $\ell = 3$, then $A_1(u,t)$ occurs when $C$ is broken as shown in the last two graphs in Figure~\ref{fig08}, 
whereas $A_2(u, t)$ never occurs; thus, $\Pr[A_1(u, t)]=\frac 16$ and $\Pr[A_2(u, t)]=0$. 
So, we always have $\frac 12 \Pr[A_1(u,t)] + \frac 14 \Pr[A_2(u,t)] \le \frac 18$.

{\em Case 3:} $c \ge 8$ and $c \equiv 0 \pmod{4}$. 
In this case, Step~3 in Figure~\ref{fig07} transforms $C$ into a collection of vertex-disjoint $3$-paths. 
Since $\Pr[A_1(u,t)] > 0$, we have $\ell = 3$.
Consequently, $A_2(u,t)$ never occurs, i.e., $\Pr[A_2(u,t)]=0$. 
So, it suffices to show that $\Pr[A_1(u,t)]=\frac 14$. 
To this end, let $e'_1$ denote the edge incident to $u$ on the $3$-path between $u$ and $t$ in $C$, 
and $e'_1, \ldots, e'_c$ denote the edges of $C$ in the order in which they appear along $C$.
Clearly, $A_1(u,t)$ occurs only if one of $e'_4, e'_8, \ldots, e'_c$ is selected as $e_1$ in Step~1 in Figure~\ref{fig07}.
Hence, $\Pr[A_1(u,t)]=\frac 14$.

{\em Case 4:} $c \ge 8$ and $c \equiv 2 \pmod{4}$. 
In this case, Step~4 in Figure~\ref{fig07} transforms $C$ into a collection of vertex-disjoint paths, each of length~$1$, $3$, or~$5$. 
Since $\Pr[A_1(u,t)] > 0$, it suffices to consider the cases where $\ell =1$, $3$, or $5$. 
To this end, we define $e'_1$, \ldots, $e'_c$ as in Case~3.

{\em Case 4.1:} $\ell = 1$. 
In this case, $e'_1 = \{ u, t \}$. 
Moreover, $A_1(u, t)$ occurs exactly when $e'_2$ is selected as $e_1$ in Step~1 and $e'_c = e_{c-1}$ is deleted from $C$ in Step~4 of Figure~\ref{fig07}. 
So, $\Pr[A_1(u,t)] = \frac 1{2c}$. 
Furthermore, $A_2(u,t)$ occurs exactly when $\{u,t\}$ is deleted from $C$ in Step~4, and hence $\Pr[A_2(u,t)]=\frac 14$ by Lemma~\ref{lemma07}. 
Therefore, $\frac 12 \Pr[A_1(u,t)] + \frac 14 \Pr[A_2(u,t)] \le \frac 18$.

{\em Case 4.2:} $\ell = 3$. 
In this case, $A_1(u,t)$ occurs exactly when the edge $e_1$ selected in Step~1 satisfies one of the following conditions:
(i) $e'_c \in \{e_1, e_5, \ldots, e_{c-9}\}$; or (ii) $e'_c = e_{c-5}$ and $e'_4 = e_{c-1}$ is deleted from $C$ in Step~4.
So, $\Pr[A_1(u,t)] = \frac {c-6}{4c}+\frac 1{2c} = \frac {c-4}{4c}$. 
Moreover, $A_2(u, t)$ occurs exactly when the edge $e_1$ selected in Step~1 is $e'_3$ and 
$e_{c-1}=e'_1$ is deleted from $C$ in Step~4. 
Thus, $\Pr[A_2(u,t)] = \frac 1{2c}$.
Therefore, $\frac 12 \Pr[A_1(u,t)] + \frac 14 \Pr[A_2(u,t)] \le \frac 18$.

{\em Case 4.3:} $\ell = 5$. In this case, 
$A_1(u,t)$ occurs exactly when the edge $e_1$ selected in Step~1 is $e'_6$ and $e'_4 = e_{c-1}$ is not deleted from $C$ in Step~4.
So, $\Pr[A_1(u,t)] = \frac 1{2c}$. 
Moreover, $A_2(u,t)$ occurs exactly when the edge $e_1$ selected in Step~1 
satisfies one of the following conditions:
\begin{itemize}
\item $e_1 = e'_2$, and $e_{c-1} = e'_c$ is deleted from $C$ in Step~4. 
\item $e_1 = e'_6$, and $e_{c-1} = e'_4$ is deleted from $C$ in Step~4. 
\item $e'_1 \in \{e_1, e_5, \ldots, e_{c-9}\}$. 
\item $e_{c-5} = e'_1$, and $e_{c-1} = e'_5$ is deleted from $C$ in Step~4.
\end{itemize}
Thus, $\Pr[A_2(u,t)] = \frac{1}{2c} + \frac{1}{2c} +\frac {c-6}{4c}+ \frac{1}{2c} = \frac 14$.
Therefore, $\frac 12 \Pr[A_1(u,t)] + \frac 14 \Pr[A_2(u,t)] \le \frac 18$.
\end{proof}

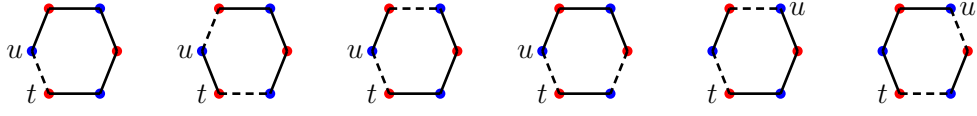
\begin{figure}[thb]
\begin{center}
\begin{tikzpicture}[scale=0.45,transform shape]
         
\foreach \x/\y in { 0/0, 5/0, 10/0, 15/0, 20/0, 25/0} {
         \fill[red] (\x, \y) circle(.15);
         \fill[blue] (\x+1.5, \y) circle(.15);
         \fill[red] (\x+2, \y-1.25) circle(.15);
         \fill[blue] (\x-0.5, \y-1.25) circle(.15);
         \fill[blue] (\x+1.5, \y-2.5) circle(.15);
         \fill[red] (\x, \y-2.5) circle(.15);
         \node[font=\fontsize{24}{6}\selectfont] at (\x-0.5, \y-2.5) {$t$};
}

\foreach \x/\y in { 0/0, 5/0, 10/0, 15/0} {
        \node[font=\fontsize{24}{6}\selectfont] at (\x-1, \y-1.25) {$u$};
        \draw[thick, line width = 1pt] (\x+1.5, \y) -- (\x+2, \y-1.25);
}

\foreach \x/\y in { 0/0, 15/0, 10/0} {
        \draw[densely dashed, line width = 1pt] (\x-0.5, \y-1.25) -- (\x, \y-2.5);
        \draw[thick, line width = 1pt] (\x-0.5, \y-1.25) -- (\x, \y);
        \draw[thick, line width = 1pt] (\x, \y-2.5) -- (\x+1.5, \y-2.5);
}
\draw[thick, line width = 1pt] (4.5, -1.25) -- (5, -2.5);
\draw[densely dashed, line width = 1pt] (4.5, -1.25) -- (5, 0);
\draw[densely dashed, line width = 1pt] (5, -2.5) -- (6.5, -2.5);

\foreach \x/\y in { 0/0, 5/0, 15/0} {
        \draw[thick, line width = 1pt] (\x, \y) -- (\x+1.5, \y);
}
\draw[densely dashed, line width = 1pt] (10, 0) -- (11.5, 0);

\foreach \x/\y in { 0/0, 10/0, 5/0} {
        \draw[thick, line width = 1pt] (\x+2, \y-1.25) -- (\x+1.5, \y-2.5);
}
\draw[densely dashed, line width = 1pt] (17, -1.25) -- (16.5, -2.5);

\foreach \x/\y in { 20/0, 25/0} {
        \node[font=\fontsize{24}{6}\selectfont] at (\x+2, \y) {$u$};
        \draw[thick, line width = 1pt] (\x, \y) -- (\x-0.5, \y-1.25);
        \draw[thick, line width = 1pt] (\x+2, \y-1.25) -- (\x+1.5, \y-2.5);
}

\draw[densely dashed, line width = 1pt] (20, 0) -- (21.5, 0);
\draw[thick, line width = 1pt] (25, 0) -- (26.5, 0);
\draw[thick, line width = 1pt] (21.5, 0) -- (22, -1.25);
\draw[densely dashed, line width = 1pt] (26.5, 0) -- (27, -1.25);
\draw[densely dashed, line width = 1pt] (19.5, -1.25) -- (20, -2.5);
\draw[thick, line width = 1pt] (24.5, -1.25) -- (25, -2.5);
\draw[thick, line width = 1pt] (20, -2.5) -- (21.5, -2.5);
\draw[densely dashed, line width = 1pt] (25, -2.5) -- (26.5, -2.5);
\end{tikzpicture}
\end{center}
\caption{All possible cases when $c = 6$ and $d_{\mc{C}_1}(u) = d_{\mc{C}_1}(t) = 1$. 
Each case occurs with probability $\frac 1{12}$.
The dashed edges are removed from $C$ by the algorithm in Figure~\ref{fig07}.
\label{fig08}}
\end{figure}

\subsection{Computing $\mc{T}_2$}

In the last subsection, we have obtained $\mc{C}_1$ from $\mc{C}$ by breaking the cycles in $\mc{C}$ independently at random. 
Recall that each edge in $X_2$ connects two vertices in different cycles of $\mc{C}$.
Let $M_2$ be the maximum-weight matching in $G[X_2]$ as defined in Lemma~\ref{lemma02}. 
We define $M'_2 = \{ \{ u, v \} \in M_2~|~d_{\mc{C}_1}(u) = d_{\mc{C}_1}(v) =1\}$. 

\begin{lemma}
\label{lemma09}
For each $e \in M_2$, $\Pr[e \in M'_2] = \frac 14$.
\end{lemma}
\begin{proof}
Suppose that $e = \{ u, v \} \in M_2$. 
By definition, $e \in M'_2$ if and only if $d_{\mc{C}_1}(u) = d_{\mc{C}_1}(v) =1$.
Since $u$ and $v$ appear in different cycles of $\mc{C}$ and 
we obtain $\mc{C}_1$ from $\mc{C}$ by breaking the cycles in $\mc{C}$ independently at random, 
$\Pr[d_{\mc{C}_1}(u) = d_{\mc{C}_1}(v) =1] = \Pr[d_{\mc{C}_1}(u) =1]\cdot \Pr[d_{\mc{C}_1}(v) =1]$. 
So, by Lemma~\ref{lemma07}, $\Pr[e \in M'_2] = \frac 14$.
\end{proof}

We construct a new graph $\mc{C}_2$ from $\mc{C}_1$ by adding the edges in $M'_2$. 
Since each connected component of $\mc{C}_1$ is an odd alternating path and 
each edge in $M'_2$ connects two vertices of degree~$1$ in $\mc{C}_1$, 
we know that each connected component of $\mc{C}_2$ is either an odd alternating path or an alternating cycle.
See Figure~\ref{fig09} for an example. 

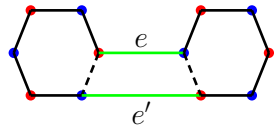
\begin{figure}[thb]
\begin{center}
\begin{tikzpicture}[scale=0.45,transform shape]
         
\foreach \x/\y in { 15/0, 20/0} {
         \fill[red] (\x, \y) circle(.15);
         \fill[blue] (\x+1.5, \y) circle(.15);
         \fill[red] (\x+2, \y-1.25) circle(.15);
         \fill[blue] (\x-0.5, \y-1.25) circle(.15);
         \fill[blue] (\x+1.5, \y-2.5) circle(.15);
         \fill[red] (\x, \y-2.5) circle(.15);
}

\foreach \x/\y in { 15/0 } {
        \draw[thick, line width = 1pt] (\x+1.5, \y) -- (\x+2, \y-1.25);
        \draw[thick, line width = 1pt] (\x-0.5, \y-1.25) -- (\x, \y-2.5);
        \draw[thick, line width = 1pt] (\x-0.5, \y-1.25) -- (\x, \y);
        \draw[thick, line width = 1pt] (\x, \y-2.5) -- (\x+1.5, \y-2.5);
        \draw[densely dashed, line width = 1pt] (\x+2, -1.25) -- (\x+1.5, -2.5);
}
\draw[thick, line width = 1pt] (15, 0) -- (16.5, 0);
\draw[thick, green, line width = 1pt] (17, -1.25) -- (19.5, -1.25);
\draw[thick, green, line width = 1pt] (16.5, -2.5) -- (20, -2.5);

\foreach \x/\y in { 20/0 } {
        \draw[thick, line width = 1pt] (\x, \y) -- (\x-0.5, \y-1.25);
        \draw[thick, line width = 1pt] (\x+2, \y-1.25) -- (\x+1.5, \y-2.5);
        \draw[thick, line width = 1pt] (\x+1.5, \y) -- (\x+2, \y-1.25);
        \draw[densely dashed, line width = 1pt] (\x-0.5, \y-1.25) -- (\x, \y-2.5);
        \draw[thick, line width = 1pt] (\x, \y-2.5) -- (\x+1.5, \y-2.5);
}

\draw[thick, line width = 1pt] (20, 0) -- (21.5, 0);
\node[font=\fontsize{24}{6}\selectfont] at (18.25, -0.9) {$e$};
\node[font=\fontsize{24}{6}\selectfont] at (18.25, -3) {$e'$};
\end{tikzpicture}
\end{center}
\caption{An example $\mc{C}$, $\mc{C}_1$, $M'_2$, and $\mc{C}_2$.
$\mc{C}$ consists of the two $6$-cycles, $\mc{C}_1$ is obtained from $\mc{C}$ by deleting the dashed edges, 
$M'_2 = \{e, e'\}$, and $\mc{C}_2$ is obtained from $\mc{C}_1$ by adding $e$ and $e'$. 
Note that $\mc{C}_2$ is a $12$-cycle.
\label{fig09}}
\end{figure}

Clearly, each cycle in $\mc{C}_2$ contains at least two edges of $M'_2$ (see Figure~\ref{fig09}). 
We obtain a new graph $\mc{C}_3$ from $\mc{C}_2$ as follows. 
For each cycle $C$ in $\mc{C}_2$, we select one edge in $M'_2 \cap E(C)$ uniformly at random 
and delete it from $\mc{C}_2$. 
For example, in Figure~\ref{fig09}, we randomly delete $e$ or $e'$ from $\mc{C}_2$ to obtain $\mc{C}_3$; in either case, $\mc{C}_3$ is an alternating $11$-path.

\begin{lemma}
\label{lemma10}
{\em \cite{HR06}}
For each $e \in M_2$, $\Pr[e \in E(\mc{C}_3)~ | ~e \in M'_2] \ge \frac 34$.
\end{lemma}
\begin{proof}
Suppose that $e = \{ u, v \}\in M_2$.
If $e \in M'_2$ and $e$ lies in a path component of $\mc{C}_2$, then $e \in E(\mc{C}_3)$ by the construction of $\mc{C}_3$. 
Hence, throughout the remainder of the proof, we may assume that whenever $e \in M'_2$, $e$ lies in a cycle component of $\mc{C}_2$. 
Note that each cycle in $\mc{C}_2$ is alternating.

Consider an arbitrary vertex $t \ne u$ such that $u$ and $t$ appear in the same cycle $C$ of $\mc{C}$.
Recall the events $A_1(u,t)$ and $A_2(u,t)$ defined immediately before Lemma~\ref{lemma08}.
We claim that $\frac 12 \Pr[A_1(u,t)~|~e \in M'_2] + \frac 14 \Pr[A_2(u,t)~|~e \in M'_2] \le \frac 14$. 
To see this, let $A_3$ and $A_4$ denote the events that $d_{\mc{C}_1}(u)=1$ and $d_{\mc{C}_1}(v)=1$, respectively.
Since $v\not\in V(C)$, but $u$ and $v$ are in $V(C)$, both $A_1(u,t)$ and $A_3$ are independent of $A_4$. 
Moreover, by Lemma~\ref{lemma07}, $\Pr[A_3] = \frac 12$ and thus, 
\[
\Pr[A_1(u,t)~|~e \in M'_2] = \Pr[A_1(u,t)~|~A_3, A_4] = \Pr[A_1(u,t)~|~A_3] = \Pr[A_1(u,t)]/\Pr[A_3] = 2\Pr[A_1(u,t)].
\]
Similarly, $\Pr[A_2(u,t)~|~e \in M'_2] = 2\Pr[A_2(u,t)]$. Now, the claim holds by Lemma~\ref{lemma08}.

Suppose that $e \in M'_2$. Then, $e$ lies in a cycle component $C'$ of $\mc{C}_2$ as assumed in the first paragraph. 
Starting from $u$, if we first traverse $e$ and further continue along $C'$, then we eventually reach the first vertex $t$ 
that lies in the same cycle component of $\mc{C}$ as $u$. 
For example, if $e$ is as shown in Figure~\ref{fig09} and its red endpoint is $u$, then $t$ is the blue endpoint of $e'$. Obviously, either $A_1(u,t)$ or $A_2(u,t)$ occurs. 
In the former case (see Figure~\ref{fig09}), we have $|E(C') \cap M'_2| \ge 2$, and in turn $\Pr[e \not\in E(\mc{C}_3)] \le \frac 12$. 
In the latter case, we have $|E(C') \cap M'_2| \ge 4$, and in turn $\Pr[e \not\in E(\mc{C}_3)] \le \frac 14$. 

By the discussion in the preceding paragraph, we have 
$\Pr[e \not\in E(\mc{C}_3)~ | ~e \in M'_2] \le \frac 12 \Pr[A_1(u,t)~|~e \in M'_2] + \frac 14 \Pr[A_2(u,t)~|~e \in M'_2]$. 
Consequently, by the above claim, $\Pr[e \not\in E(\mc{C}_3)~ | ~e \in M'_2]\le\frac 14$. 
This completes the proof.
\end{proof}

Clearly, $V(\mc{C}_3) \cup V(\mc{P}) = V$, and the union of $\mc{C}_3$ and $\mc{P}$ is a collection of vertex-disjoint alternating paths. 
Therefore, we compute the second perfect fair-triangle packing $\mc{T}_2$ from $\mc{C}_3$ and $\mc{P}$ in $O(n)$ time as shown in Lemma~\ref{lemma05}. 
The entire procedure is summarized in Figure~\ref{fig10}.

\begin{figure}[htb]
\begin{center}
\framebox{
\begin{minipage}{5.5in}
{\em Algorithm for computing $\mc{T}_2$}:\\
Input: $G$, $\mc{C}$, and $\mc{P}$;
\begin{itemize}
\parskip=0pt
\item[1.]
	Construct $\mc{C}_1$ from $\mc{C}$ by processing each cycle $C$ in $\mc{C}$ independently 
	using the algorithm in Figure~\ref{fig07}.
         
\item[2.]
         Compute a maximum-weight matching $M_2$ in $G[X_2]$.

\item[3.]
         Let $M'_2 = \{ \{ u, v \} \in M_2~|~d_{\mc{C}_1}(u) = d_{\mc{C}_1}(v) =1\}$, and 
	  construct $\mc{C}_2$ from $\mc{C}_1$ by adding the edges in $M'_2$.
         
\item[4.] 
	  Constrcut $\mc{C}_3$ from $\mc{C}_2$ by processing each cycle component $C$ of $\mc{C}_2$ 
	  independently as follows: 
	  select an $e\in E(C) \cap M'_2$ uniformly at random and delete it. 
	({\em Comment:} The union of $\mc{C}_3$ and $\mc{P}$ is a collection of 
	vertex-disjoint alternating paths with $V(\mc{C}_3) \cup V(\mc{P}) = V$.) 
         
\item[5.]
	 Construct a perfect fair-triangle packing $\mc{T}_2$ from the union of $\mc{C}_3$ and $\mc{P}$ 
	as shown in Lemma~\ref{lemma05}. 
\item[6.] Return $\mc{T}_2$.
\end{itemize}
\end{minipage}}
\end{center}
\caption{Computing the second perfect fair-triangle packing $\mc{T}_2$.} \label{fig10}
\end{figure}

\begin{lemma}
\label{lemma11}
The algorithm in Figure~\ref{fig10} computes a perfect fair-triangle packing $\mc{T}_2$ in $O(n^{2.5})$ time 
such that $\mathbb{E}[w(\mc{T}_2)] \ge \frac 12 w(\mc{C}) + \frac 23 w(\mc{P}) + \frac 1{16} w(X^*_2)$.
\end{lemma}
\begin{proof}
The running time of the algorithm is dominated by Step~2. 
Since $|V|=3n$ and $|X_2| \le 9n^2$, Step~2 takes $O(n^{2.5})$ time~\cite{MV80}. 
Since $\mathbb{E}[w(\mc{T}_2)] \ge \frac 23 \mathbb{E}[w(\mc{C}_3)] + \frac 23 w(\mc{P})$
by Lemma~\ref{lemma05}, 
it remains to show that $\mathbb{E}[w(\mc{C}_3)] \ge \frac 34 w(\mc{C}) + \frac 3{32} w(X^*_2)$.

Clearly, $E(\mc{C}_1)$ and $E(\mc{C}_3) \cap M_2$ form a partition of $E(\mc{C}_3)$.
By Lemma~\ref{lemma07}, $\mathbb{E}[w(\mc{C}_1)] = \frac 34 w(\mc{C})$. 
Moreover, for each $e \in M_2$, $\Pr[e \in M'_2] = \frac 14$ by Lemma~\ref{lemma09}, and hence 
Lemma~\ref{lemma10} implies 
\[
\Pr[e \in E(\mc{C}_3)] = \Pr[e \in E(\mc{C}_3)~|~e \in M'_2] \cdot \Pr[e \in M'_2] \ge \frac 3{16}.
\]
By Lemma~\ref{lemma02}, $w(M_2) \ge \frac 12 w(X^*_2)$. 
Therefore, $\mathbb{E}[w(E(\mc{C}_3) \cap M_2)] \ge \frac 3{32} w(X^*_2)$.
This completes the proof.
\end{proof}

\subsection{Computing $\mc{T}_3$}

We continue to use $\mc{C}_1$ constructed in Step~1 of the algorithm in Figure~\ref{fig10}. 
We next show how to use $\mc{P}$ to compute another collection of vertex-disjoint alternating paths, 
each of length $1$ or $2$.

Recall Notation~\ref{nota01}. 
Let $r_p=|V_{\rm r} \cap V(\mc{P})|$ be the number of red vertices in $\mc{P}$.
Similarly, we define $b_p=|V_{\rm b} \cap V(\mc{P})|$. 
We start by analyzing $r_p$ and $b_p$. 
Since  every cycle component of $\mc{C}$ is alternating and has even length, $|V_{\rm r} \cap V(\mc{C})| = |V_{\rm b} \cap V(\mc{C})|$.
Note that
the net difference between the number of blue and red vertices in $G$ is $(3n-r)-r$.
Therefore, we have $b_p - r_p=3n-2r$. 
Moreover, each alternating path contains at most one more blue vertex than red vertices. 
It follows that $\mc{P}$ contains at least $3n-2r$ alternating paths. 
Since every alternating path contains at least one red vertex, we conclude that $r_p \ge 3n-2r$.

Recall Notation~\ref{nota03}. 
For simplicity, let $X'_1=\{\{ u, v \}\in X_1~|~u, v \in V(\mc{P})\}$. 
Clearly, $X_5 \cup X'_1$ contains every bichromatic edge whose endpoints lie in $\mc{P}$.
We compute a maximum-weight matching $M_{\mc{P}}$ in $G[X_5 \cup X'_1]$. 
Since $b_p \ge r_p$, $M_{\mc{P}}$ leaves at least as many blue vertices in $\mc{P}$ unmatched as red vertices. 
Thus, if $|M_{\mc{P}}| < r_p$, we augment $M_{\mc{P}}$ by matching each unmatched red vertex in $\mc{P}$ to a distinct unmatched blue vertex in $\mc{P}$. 
Consequently, $|M_{\mc{P}}| = r_p$, and $M_{\mc{P}}$ contains a maximum-weight matching in $G[X_5 \cup X'_1]$ as a subset. 

\begin{lemma}
\label{lemma12}
$w(M_{\mc{P}}) \ge \frac 12 \max \{ w(X^*_5), w(\mc{P}) \}$.
\end{lemma}
\begin{proof}
Recall that $X^*_5 \subseteq X_5 \subseteq X_5 \cup X'_1$ and that $M_{\mc{P}}$ contains a maximum-weight matching in $G[X_5 \cup X'_1]$ as a subset. 
So, $w(M_{\mc{P}}) \ge \frac 12 w(X^*_5)$ by Lemma~\ref{lemma02}.

Since $\mc{P}$ is a collection of vertex-disjoint paths and the edge set of any path can be partitioned into two matchings, 
$E(\mc{P})$ contains a matching of weight at least $\frac 12 w(\mc{P})$.
Clearly, $E(\mc{P}) \subseteq X'_1 \subseteq X_5 \cup X'_1$. 
Finally, as $M_{\mc{P}}$ contains a maximum-weight matching in $G[X_5 \cup X'_1]$ as a subset, we obtain $w(M_{\mc{P}}) \ge \frac 12 w(\mc{P})$.
This completes the proof.
\end{proof}

Since $|M_{\mc{P}}| = r_p$, exactly $3n-2r$ blue vertices in $V(\mc{P})$ remain unmatched by $M_{\mc{P}}$.
Using $r_p \ge 3n-2r$, we connect each unmatched blue vertex of $\mc{P}$ to the red endpoint of a distinct edge in $M_{\mc{P}}$. 
This produces $3n-2r$ alternating $2$-paths (each having two blue endpoints), and leaves exactly $r_p-3n+2r$ edges of $M_{\mc{P}}$ as standalone $1$-paths.
Let $\mc{P}_1$ denote this collection of $2$-paths and $1$-paths. 
Clearly, $V(\mc{P}_1) = V(\mc{P})$.
See Figure~\ref{fig11} for an illustration. 

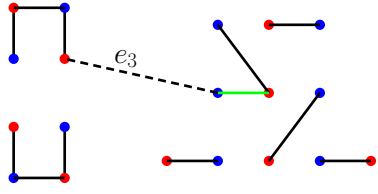
\begin{figure}[thb]
\begin{center}
\begin{tikzpicture}[scale=0.45,transform shape]
         
\foreach \x/\y in { 0/0, 0/3.5 } {
         \draw[thick, line width = 1pt] (\x, \y) -- (\x, \y-1.5);
         \draw[thick, line width = 1pt] (\x+1.5, \y) -- (\x+1.5, \y-1.5);
         \fill[red] (\x, \y) circle(.15);
         \fill[blue] (\x+1.5, \y) circle(.15);
         \fill[red] (\x+1.5, \y-1.5) circle(.15);
         \fill[blue] (\x, \y-1.5) circle(.15);
}

\draw[thick, line width = 1pt] (0, -1.5) -- (1.5, -1.5);
\draw[thick, line width = 1pt] (0, 3.5) -- (1.5, 3.5);

\foreach \x/\y in { 6/3, 6/1, 6/-1 } {
         \fill[blue] (\x, \y) circle(.15);
         \fill[red] (\x+1.5, \y) circle(.15);
         \fill[blue] (\x+3, \y) circle(.15);
}
\fill[red] (4.5, -1) circle(.15);
\fill[red] (10.5, -1) circle(.15);

\foreach \x/\y in { 7.5/3, 4.5/-1, 9/-1 } {
         \draw[thick, line width = 1pt] (\x, \y) -- (\x+1.5, \y);
}
\draw[thick, line width = 1pt] (6, 3) -- (7.5, 1);
\draw[thick, line width = 1pt] (7.5, -1) -- (9, 1);
\draw[thick, green, line width = 1pt] (6, 1) -- (7.5, 1);

\draw[densely dashed, line width = 1pt] (6, 1) -- (1.5, 2);
\node[font=\fontsize{24}{6}\selectfont] at (3.3, 2) {$e_3$};
\end{tikzpicture}
\end{center}
\caption{An illustration of the construction of $\mc{P}_1$ and $M'_3$. 
The two $3$-paths on the left are obtained by breaking the cycles in Figure~\ref{fig04}, while the remaining vertices lie in $\mc{P}$. 
The five non-green solid edges on the right belong to $M_{\mc{P}}$. 
In the construction of $\mc{P}_1$, the green edge is used to connect the blue endpoint of $e_3$ 
to the red endpoint of an edge in $M_{\mc{P}}$. 
Thus, the solid $2$-path on the right is a path component of $\mc{P}_1$. 
Assuming $e_3 \in M_3$, $e_3 \in M'_3$ holds.
\label{fig11}}
\end{figure}

Clearly, $V(\mc{C}) = V(\mc{C}_1)$ and $V(\mc{P}) = V(\mc{P}_1)$. 
Thus, each edge in $X_3$ connects a red vertex of $\mc{C}_1$ to a blue vertex of $\mc{P}_1$. 
Morever, every blue vertex $v$ in $\mc{P}_1$ satisfies $d_{\mc{P}_1}(v) = 1$. 
Let $M_3$ be as defined in Lemma~\ref{lemma02}. 
We define $M'_3 =\{\{u,v\}\in M_3~|~u\in V(\mc{C}_1)$ is red and $d_{\mc{C}_1}(u) = 1\}$ (see Figure~\ref{fig11} for an illustration). 
Let $\mc{F}_1$ denote the graph obtained from the union of $\mc{C}_1$ and $\mc{P}_1$ by adding the edges in $M'_3$.
That is, $\mc{F}_1 = G[E(\mc{C}_1) \cup E(\mc{P}_1) \cup M'_3]$.

Let $P$ be an arbitrary connected component of $\mc{F}_1$.
By the construction of $\mc{F}_1$, $P$ is clearly a path or cycle. 
We claim that, in fact, $P$ is an alternating path. 
If $E(P) \cap M'_3 = \emptyset$, then the claim holds because $P$ is a connected component of either $\mc{C}_1$ or 
$\mc{P}_1$; 
in either case, $P$ is an alternating path. 
We next consider the case where  $E(P) \cap M'_3 \ne \emptyset$.
In this case, $P$ contains an edge $\{ u, v \} \in M'_3$, where $u \in V(\mc{C}_1)$ is red. 
Let $K$ be the connected component of $\mc{C}_1$ containing $u$. 
By the algorithm in Figure~\ref{fig07}, $K$ is an odd alternating path, and $u$ is an endpoint of $K$. 
Since $u$ is red, the other endpoint $t$ of $K$ is blue. 
Hence, no edge of $M'_3$ is incident to $t$. 
Consequently, $P$ is a path and is also alternating. 
See the $6$-path containing the green edge and $e_3$ in Figure~\ref{fig11} for an illustration. 
This completes the proof of the claim.

By the claim proved in the preceding paragraph, each connected component of $\mc{F}_1$ is an alternating path. 
Since $V(\mc{F}_1) = V$, we can construct a perfect fair-triangle packing $\mc{T}_3$ from $\mc{F}_1$ in $O(n)$ time as shown in Lemma~\ref{lemma05}. 
The entire algorithm is summarized in Figure~\ref{fig12}.

\begin{figure}[htb]
\begin{center}
\framebox{
\begin{minipage}{5.5in}
{\em Algorithm for computing $\mc{T}_3$}:\\
Input: $G$, $\mc{C}$, and $\mc{P}$;
\begin{itemize}
\parskip=0pt
\item[1.]
	Compute a maximum-weight matching $M_{\mc{P}}$ in $G[X_5 \cup X'_1]$, where 
	$X'_1$ consists of all edges $\{ u, v \} \in X_1$ with $u \in V(\mc{P})$ and $v \in V(\mc{P})$.
         
\item[2.]
	Let $r_p$ denote the number of red vertices in $\mc{P}$. 
	If $|M_{\mc{P}}| < r_p$, augment $M_{\mc{P}}$ by matching each unmatched red vertex in $\mc{P}$ to
	a distinct unmatched blue vertex in $\mc{P}$. 
	({\em Comment:} After this step, $|M_{\mc{P}}| = r_p$ and exactly $3n-2r$ blue vertices in $V(\mc{P})$ remain unmatched by $M_{\mc{P}}$.)
         
\item[3.]
	Construct a graph $\mc{P}_1$ from the graph $(V(\mc{P}), M_{\mc{P}})$ by connecting each unmatched blue vertex 
	of $\mc{P}$ to the red endpoint of a distinct edge in $M_{\mc{P}}$. 

\item[4.]
	Compute a maximum-weight matching $M_3$ in $G[X_3]$.
         
\item[5.] 
	Let $M'_3 =\{\{ u, v \}\in M_3~|~u\in V(\mc{C}_1)$ is red and $d_{\mc{C}_1}(u) = 1\}$, 
	where $\mc{C}_1$ is computed in the first step in Figure~\ref{fig10}.
         
\item[6.]
	Let $\mc{F}_1 = G[E(\mc{C}_1) \cup E(\mc{P}_1) \cup M'_3]$. 
	 ({\em Comment:} $V(\mc{F}_1) = V$ and $\mc{F}_1$ is a collection of vertex-disjoint alternating paths.)

\item[7.]
	Construct a perfect fair-triangle packing $\mc{T}_3$ from $\mc{F}_1$ as shown in Lemma~\ref{lemma05}. 

\item[8.] Return $\mc{T}_3$. 
\end{itemize}
\end{minipage}}
\end{center}
\caption{Computing the third fair-triangle packing $\mc{T}_3$.} \label{fig12}
\end{figure}

\begin{lemma}
\label{lemma13}
The algorithm in Figure~\ref{fig12} computes a perfect fair-triangle packing $\mc{T}_3$ in $O(n^{2.5})$ time 
such that $\mathbb{E}[w(\mc{T}_3)] \ge \frac 12 w(\mc{C}) + \frac 16 w(X^*_3) + \frac 13 \max \{ w(X^*_5), w(\mc{P}) \}$.
\end{lemma}
\begin{proof}
The running time of the algorithm is dominated by the computation of $M_{\mc{P}}$ and $M_3$. 
Thus, it takes $O(n^{2.5})$ time. 
It remains to estimate $\mathbb{E}[w(\mc{T}_3)]$. 

Consider an arbitrary edge $e = \{ u, v \} \in M_3$ such that $u \in V(\mc{C}_1)$ and $u$ is red.
By Lemma~\ref{lemma07}, $\Pr[d_{\mc{C}_1}(u)=1] = \frac 12$ and hence $\Pr[e \in M'_3] = \frac 12$. 
Moreover, $w(M_3) \ge \frac 12 w(X^*_3)$ by Lemma~\ref{lemma02}. 
Therefore, $\mathbb{E}[w(M'_3)] \ge \frac 12 w(M_3) \ge \frac 14 w(X^*_3)$.

Since $E(\mc{F}_1) = E(\mc{C}_1) \cup E(\mc{P}_1) \cup M'_3$, Lemma~\ref{lemma05} implies that 
$\mathbb{E}[w(\mc{T}_3)] \ge \frac 23 \mathbb{E}[w(\mc{C}_1)] + \frac 23 \mathbb{E}[w(M'_3)] + \frac 23 w(\mc{P}_1)$.
Moreover, $\mathbb{E}[w(\mc{C}_1)] = \frac 34 w(\mc{C})$ by Lemma~\ref{lemma07}. 
Finally, Lemma~\ref{lemma12}, together with $M_{\mc{P}} \subseteq E(\mc{P}_1)$, implies that 
 $w(\mc{P}_1) \ge \frac 12 \max \{ w(X^*_5), w(\mc{P}) \}$. 
The lemma now follows from the last inequality in the preceding paragraph.
\end{proof}

\subsection{Computing $\mc{T}_4$}

Recall that the construction of $\mc{T}_3$ involved the edges in $X_3$. 
Similarly, the construction of $\mc{T}_4$ will involve the edges in $X_4$. 

To compute $\mc{T}_4$, we first obtain $M_{\mc{P}}$ via Steps~1 and~2 in Figure~\ref{fig12}.
Since $|M_{\mc{P}}| = r_p \ge 3n-2r$ and exactly $3n-2r$ blue vertices in $V(\mc{P})$ remain unmatched by $M_{\mc{P}}$, 
we connect each unmatched blue vertex of $\mc{P}$ to the blue endpoint of a distinct edge in $M_{\mc{P}}$. 
This produces $3n-2r$ special $2$-paths (see Definition~\ref{def03}), and leaves exactly $r_p-3n+2r$ edges of $M_{\mc{P}}$, as standalone $1$-paths.
Let $\mc{P}_2$ denote this collection of $1$-paths and special $2$-paths. 
See Figure~\ref{fig13} for an illustration.

\begin{figure}[thb]
\begin{center}
\begin{tikzpicture}[scale=0.45,transform shape]
         
\foreach \x/\y in { 0/0, 0/3.5 } {
         \draw[thick, line width = 1pt] (\x, \y) -- (\x, \y-1.5);
         \draw[thick, line width = 1pt] (\x+1.5, \y) -- (\x+1.5, \y-1.5);
         \fill[red] (\x, \y) circle(.15);
         \fill[blue] (\x+1.5, \y) circle(.15);
         \fill[red] (\x+1.5, \y-1.5) circle(.15);
         \fill[blue] (\x, \y-1.5) circle(.15);
}

\draw[thick, line width = 1pt] (0, -1.5) -- (1.5, -1.5);
\draw[thick, line width = 1pt] (0, 3.5) -- (1.5, 3.5);

\foreach \x/\y in { 6/3, 6/1, 6/-1 } {
         \fill[blue] (\x, \y) circle(.15);
         \fill[red] (\x+1.5, \y) circle(.15);
         \fill[blue] (\x+3, \y) circle(.15);
}
\fill[red] (4.5, -1) circle(.15);
\fill[red] (10.5, -1) circle(.15);

\foreach \x/\y in { 7.5/3, 4.5/-1, 9/-1 } {
         \draw[thick, line width = 1pt] (\x, \y) -- (\x+1.5, \y);
}
\draw[thick, line width = 1pt] (6, 3) -- (7.5, 1);
\draw[thick, line width = 1pt] (7.5, -1) -- (9, 1);
\draw[thick, green, line width = 1pt] (6, 1) -- (6, -1);

\draw[densely dashed, line width = 1pt] (4.5, -1) -- (1.5, 0);
\node[font=\fontsize{24}{6}\selectfont] at (3, -1) {$e_4$};
\end{tikzpicture}
\end{center}
\caption{An illustration of the construction of $\mc{P}_2$ and $M'_4$.
The two $3$-paths on the left are obtained by breaking the cycles in Figure~\ref{fig04}, 
while the remaining vertices lie in $\mc{P}$. 
The five non-green solid edges on the right belong to $M_{\mc{P}}$. 
In the construction of $\mc{P}_2$, the green edge is used to extend an edge in $M_{\mc{P}}$ into a special $2$-path. Thus, the solid $2$-path on the right is a path component of $\mc{P}_2$. 
Assuming $e_4 \in M_4$, $e_4\in M'_4$ holds.
\label{fig13}}
\end{figure}
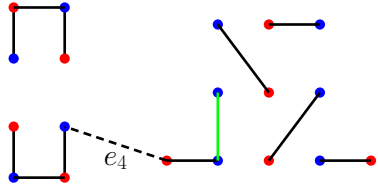

Let $M_4$ be as defined in Lemma~\ref{lemma02}. 
We define $M'_4 =\{\{u,v\}\in M_4~|~u\in V(\mc{C}_1)$ is blue and $d_{\mc{C}_1}(u) = 1\}$ 
(see Figure~\ref{fig13} for an illustration). 
Let $\mc{F}_2$ denote the graph obtained from the union of $\mc{C}_1$ and $\mc{P}_2$ by adding the edges in $M'_4$.
That is, $\mc{F}_2 = G[E(\mc{C}_1) \cup E(\mc{P}_2) \cup M'_4]$.

Let $P$ be an arbitrary connected component of $\mc{F}_2$.
We claim that $P$ is either an odd alternating path or a special path.
If $E(P) \cap M'_4 = \emptyset$, then the claim holds because $P$ is a connected component of either $\mc{C}_1$ or $\mc{P}_2$; 
in the former case, $P$ is an odd alternating path, whereas in the latter case, $P$ is either an alternating $1$-path or a special $2$-path.
We next consider the case where $E(P) \cap M'_4 \ne \emptyset$. 
In this case, $P$ contains an edge $e = \{ u, v \} \in M'_4$, where $u\in V(\mc{C}_1)$ is blue and $v\in V(\mc{P}_2)$ is red.
Let $K_u$ (respectively, $K_v$) denote the path component of $\mc{C}_1$ (respectively, $\mc{P}_2$) containing $u$ (respectively, $v$).  
By the algorithm in Figure~\ref{fig07}, $K_u$ is an odd alternating path, and $u$ is an edpoint of $K_u$. 
Since $u$ is blue, the other endpoint $t$ of $K_u$ is red. 
Hence, no edge in $M'_4$ is incident to $t$. 
Recall that $K_v$ is either an alternating $1$-path or a special $2$-path. 
In either case, its other endpoint is blue; hence no edge in $M'_4$ is incident to it.
It follows that $E(P) \cap M'_4 = \{ e \}$. 
Clearly, $P$ is an odd alternating path if $K_v$ is an alternating $1$-path, and a special path if $K_v$ is a special $2$-path. 
Thus, the claim holds. 

By the claim proved in the preceding paragraph, each connected component of $\mc{F}_2$ is either an odd alternating path or a special path. 
Since $V(\mc{F}_2) = V$, we can construct a perfect fair-triangle packing $\mc{T}_4$ from $\mc{F}_2$ in $O(n)$ time as shown in Lemma~\ref{lemma06}. 
The entire algorithm is summarized in Figure~\ref{fig14}.

\begin{figure}[htb]
\begin{center}
\framebox{
\begin{minipage}{5.5in}
{\em Algorithm for computing $\mc{T}_4$}:\\
Input: $G$, $\mc{C}$, and $\mc{P}$;
\begin{itemize}
\parskip=0pt
\item[1.]
	Perform the first two steps in Figure~\ref{fig12} to compute $M_{\mc{P}}$.
         
\item[2.]
      Construct a graph $\mc{P}_2$ from the graph $(V(\mc{P}), M_{\mc{P}})$ by connecting each unmatched blue vertex 
	of $\mc{P}$ to the blue endpoint of a distinct edge in $M_{\mc{P}}$. 

\item[3.]
	Compute a maximum-weight matching $M_4$ in $G[X_4]$.
         
\item[4.] 
	Let $M'_4 =\{\{u,v\}\in M_4~|~u\in V(\mc{C}_1)$ is blue and $d_{\mc{C}_1}(u) = 1\}$, where $\mc{C}_1$ is computed 
	by the first step in Figure~\ref{fig10}.
         
\item[5.]
	Let $\mc{F}_2 = G[E(\mc{C}_1) \cup E(\mc{P}_2) \cup M'_4]$. ({\em Comment:} $V(\mc{F}_2) = V$ and
	each connected component of $\mc{F}_2$ is either an odd alternating path or a special path.) 

\item[6.] 
         Construct a perfect fair-triangle packing $\mc{T}_4$ from $\mc{F}_2$ as shown in Lemma~\ref{lemma06}.

\item[7.] Return $\mc{T}_4$.
\end{itemize}
\end{minipage}}
\end{center}
\caption{Computing the fourth perfect fair-triangle packing $\mc{T}_4$.} \label{fig14}
\end{figure}

\begin{lemma}
\label{lemma14}
The algorithm in Figure~\ref{fig14} computes a perfect fair-triangle packing $\mc{T}_4$ in $O(n^{2.5})$ time 
such that $\mathbb{E}[w(\mc{T}_4)] \ge \frac 12 w(\mc{C}) + \frac 16 w(X^*_4) + \frac 13 \max \{ w(X^*_5), w(\mc{P}) \}$.
\end{lemma}
\begin{proof}
The running time of the algorithm is dominated by the computation of $M_{\mc{P}}$ and $M_4$. 
Thus, it takes $O(n^{2.5})$ time. 
It remains to estimate $\mathbb{E}[w(\mc{T}_4)]$. 

Consider an arbitrary edge $e = \{ u, v \} \in M_4$ such that $u \in V(\mc{C}_1)$ is blue.
By Lemma~\ref{lemma07}, $\Pr[d_{\mc{C}_1}(u) = 1] = \frac 12$ and hence $\Pr[e \in M'_4] = \frac 12$. 
Moreover, $w(M_4) \ge \frac 12 w(X^*_4)$ Lemma~\ref{lemma02}. 
Therefore, $\mathbb{E}[w(M'_4)] \ge \frac 12 w(M_4) \ge \frac 14 w(X^*_4)$.

Since $E(\mc{F}_2) = E(\mc{C}_1) \cup E(\mc{P}_2) \cup M'_4$, Lemma~\ref{lemma06} implies that 
$\mathbb{E}[w(\mc{T}_4)] \ge \frac 23 \mathbb{E}[w(\mc{C}_1)] + \frac 23 \mathbb{E}[w(M'_4)] + \frac 23 w(\mc{P}_2)$. 
Moreover, $\mathbb{E}[w(\mc{C}_1)] = \frac 34 w(\mc{C})$ by Lemma~\ref{lemma07}. 
Finally, by Lemma~\ref{lemma12}, together with $M_{\mc{P}} \subseteq \mc{P}_2$, implies that 
$w(\mc{P}_2) \ge \frac 12 \max \{ w(X^*_5), w(\mc{P}) \}$.
The lemma now follows from the last inequality in the preceding paragraph. 
\end{proof}

\subsection{The complete algorithm}

Given a bicolored, edge-weighted complete graph $G$, our algorithm, denoted by {\sc Approx2}, 
constructs five perfect fair-triangle packings $\mc{T}_0, \ldots, \mc{T}_4$ (shown in 
Figures~\ref{fig03}, \ref{fig05}, \ref{fig10}, \ref{fig12} and \ref{fig14}, respectively),  
and outputs the one with maximum weight.

\begin{theorem}
\label{thm03}
Given any fixed $\epsilon > 0$, {\sc Approx2} achieves an expected approximation ratio of $\frac {16}{47}-\epsilon$
for the PFTP problem in $O(n^4)$ time.
\end{theorem}
\begin{proof}
Computing the initial $\mc{F}$ and preprocessing it can be done in $O(n^3)$ time.
The running time of {\sc Approx2} follows from Theorem~\ref{thm02} and Lemmas~\ref{lemma04}, \ref{lemma11}, \ref{lemma13} and \ref{lemma14}.
Let $\mc{T}$ be the perfect fair-triangle packing returned by {\sc Approx2} on input $G$. We next analyze $\mathbb{E}[\mc{T}]$.

As observed in the proof of Theorem~\ref{thm02}, $w(\mc{T}_0) \ge w(R^*) + w(B^*)$.
Obviously, $w(X^*) = \sum_{i=1}^5 w(X^*_i)$. Moreover, by P2 in Remark~\ref{remark01}, 
$w(\mc{C})+w(\mc{P}) = w(\mc{F}) \ge (1-\epsilon)w(X^*)$. Thus, 
by Lemmas~\ref{lemma04}, \ref{lemma11}, \ref{lemma13} and \ref{lemma14} again, we have
\begin{eqnarray*}
\mathbb{E}[\mc{T}] &  =  &  \max \{ w(\mc{T}_0), w(\mc{T}_1), \mathbb{E}(\mc{T}_2),  \mathbb{E}(\mc{T}_3), \mathbb{E}(\mc{T}_4)\} \\
& \ge & \frac {16}{47} w(\mc{T}_0) + \frac 1{47} w(\mc{T}_1) + \frac {18}{47} \mathbb{E}[w(\mc{T}_2)] + \frac 6{47} \mathbb{E}[w(\mc{T}_3)]  + \frac 6{47} \mathbb{E}[w(\mc{T}_4)]  \\
& = & \frac {16}{47} w(\mc{T}_0) + \frac {15}{47} w(\mc{C}) + \frac {12}{47} w(\mc{P}) + 
\frac 4{47} \max \{ w(X^*_5), w(\mc{P}) \} + \frac 1{47} \sum_{i \in \{ 1, 3, 4 \}} w(X^*_i)+\frac 9{376} w(X^*_2) \\
& \ge & \frac {16}{47} w(\mc{T}_0) + \frac {15}{47} w(\mc{C}) + \frac {12}{47} w(\mc{P}) + \frac 3{47} \max \{ w(X^*_5), w(\mc{P}) \} + \frac 1{47} \sum_{i=1}^5 w(X^*_i) \\
& \ge & \frac {16}{47} w(\mc{T}_0) + \frac {15}{47} w(\mc{C}) + \frac {15}{47} w(\mc{P}) + \frac 1{47} w(X^*) \\
& \ge & \frac {16}{47} w(R^* \cup B^*) + \left(\frac {16}{47} - \frac {15}{47} \epsilon\right) w(X^*)
\ge \left(\frac {16}{47} - \epsilon\right)w(\mc{T}^*).
\end{eqnarray*}
This completes the proof.
\end{proof}

\section{Concluding remarks}\label{sec:conclude}

We have introduced the perfect fair-triangle packing problem (PFTP) 
and designed two approximation algorithms for it. 
The first is a deterministic, matching-based algorithm that runs in $O(n^3)$ time and achieves an approximation ratio of $\frac 13$. 
The second is a randomized algorithm that runs in $O(n^4)$ time and achieves an expected approximation ratio of $\frac {16}{47}-\epsilon$ for any fixed small constant $\epsilon>0$. 

We highlight several natural directions for future research.
First, while existing algorithms for W$3$SP cannot be straightforwardly applied to PFTP (as discussed in Section~\ref{sec1.2}), 
they may be still be effective as subroutines for designing better approximation algorithms.
Second, it would be interesting to study the metric version of PFTP (as in~\cite{CCL21, ZX24}), where edge weights satisfy the triangle inequality.
Third, the fairness criterion can be extended to broader cycle and path packing problems.
For example, let $k \ge 3$ be an integer, and $G$ be a bicolored, edge-weighted complete graph with $|V(G)| = kn$.
A $k$-cycle or $(k-1)$-path in $G$ is called {\em fair} if it contains at least one vertex of each color.
The {\em perfect fair-$k$-cycle} (respectively, {\em perfect fair-$(k-1)$-path}) {\em packing} problem asks for 
a collection of $n$ vertex-disjoint fair $k$-cycles (respectively, $(k-1)$-paths) with maximum total edge weight.
Finally, generalizing beyond two vertex colors to multi-colored variants offers another rich direction for future study.




\bibliography{PFTP.bib}

\appendix

\section{A table for notations}

\begin{table}[H]
\caption{Important notations and their meanings}
\label{tab01}
\centering
\begin{tabular}{|c|c|c|c|c|c|c|c|c|c|c|c|c|}
\hline
Notations & Their meanings  \\
\hline
$X$ & The set of bichromatic edges in $G$  \\
\hline
$G_{\rm x}$ & The spanning subgraph $G[X]$ of $G$  \\
\hline
$\mc{F}$ & Initially, a maximum-weight $[1, 2]$-factor in $G_{\rm x}$ \\
\hline
$\mc{C}$ & The set of cycles in $\mc{F}$ after the prepocessing in Section~\ref{subsec:prep} \\
\hline
$\mc{P}$ & The set of paths in $\mc{F}$ after the prepocessing in Section~\ref{subsec:prep} \\
\hline
$X_i$ & The set of bichromatic edges of Type $i$ ($1\le i\le 5$) in $G$  \\
\hline
$M_i$ & A maximum-weight matching in $G[X_i]$, where $i\in\{1,2,\ldots,5\}$  \\
\hline
$\mc{C}_1$ & The graph obtained from $\mc{C}$ by breaking its cycles independently at random \\
\hline
$M'_2$ &  The set of edges $\{ u, v \} \in M_2$ such that $d_{\mc{C}_1}(u)=d_{\mc{C}_1}(v)=1$ \\
\hline
$\mc{C}_2$ &  The graph obtained by adding $M'_2$ to $\mc{C}_1$, i.e., $(V(\mc{C}), E(\mc{C}_1) \cup M'_2)$  \\
\hline
$\mc{C}_3$ &  The graph obtained by randomly deleting one edge of $M'_2$ from each cycle in $\mc{C}_2$  \\
\hline
$X'_1$ &  The set of edges $\{ u, v \} \in X_1$ with $u \in V(\mc{P})$ and $v \in V(\mc{P})$ \\
\hline
$M_{\mc{P}}$ &  A maximum-weight matching in $G[X'_1 \cup X_5]$ \\
\hline
$\mc{P}_1$ &   A collection of alternating $1$- or $2$-paths constructed from $M_{\mc{P}}$ \\ 
\hline
$M'_3$ &  The set of edges $\{ u, v \} \in M_3$ such that $d_{\mc{C}_1}(u) = 1$ and $u$ is red \\
\hline
$\mc{F}_1$ &  $G[E(\mc{C}_1) \cup E(\mc{P}_1) \cup M'_3]$ \\
\hline
$\mc{P}_2$ &  A collection of alternating $1$-paths and special $2$-paths constructed from $M_{\mc{P}}$ \\
\hline
$M'_4$ &  The set of edges $\{ u, v \} \in M_4$ such that $d_{\mc{C}_1}(u) = 1$ and $u$ is blue \\
\hline
$\mc{F}_2$ &  $G[E(\mc{C}_1) \cup E(\mc{P}_2) \cup M'_4]$ \\
\hline
\end{tabular}
\end{table}

\end{document}